\title{Persuading part of an audience}
\author{
	Bruno Salcedo%
	\thanks{
		Department of Economics, Western  University,
		{\href{http://www.brunosalcedo.com/}{{brunosalcedo.com}}},
		{\href{mailto:bsalcedo@uwo.ca}{{bsalcedo@uwo.ca}}}.
		}
	\thanks{ I am grateful to Yi Chen for introducing me to the work of \cite{lipnowski2017},
		and to the Economics faculty at Western University for their support and guidance.
		}
}
\date{\enspace\today}
\documentclass[12pt]{article}
\usepackage[english]{babel}    
\usepackage[ansinew]{inputenc} 
\usepackage[T1]{fontenc}       
\usepackage{lmodern,microtype} 
\usepackage{amsmath,amsthm,amssymb,amsfonts} 
\usepackage{dsfont,mathrsfs,ushort,xfrac,bm} 
\usepackage{titlesec,titling}  
\usepackage[nohead]{geometry}  
\usepackage{setspace,caption}  
\usepackage{enumitem,booktabs} 
\usepackage[comma]{natbib} 
\usepackage{pstricks,pst-all}   
\usepackage{pst-plot,pst-node,pst-intersect}  
\usepackage{sgame}           
\usepackage{hyperref,breakurl}
\setenumerate{label=\small(\roman*)}
\hypersetup{breaklinks=true,colorlinks=true,pdfnewwindow=true,pdfstartview=FitH,%
	pdftitle="Choice interdependence",pdfauthor="Bruno Salcedo",%
	urlcolor=blue!95!red!50!black,citecolor=blue!95!red!50!black,linkcolor=red!90!black}
\DeclareCaptionFont{fancy}{\bfseries\sffamily} 
\gamemathtrue
\allowdisplaybreaks

\newcommand{\zibri}{\"{U}///}
\newcommand{\quokkito}{\"{q}}
\AtEndDocument{\vfill\centering\footnotesize\textcolor{black!14}{\quokkito \quad \zibri }}
\providecommand{\psreset}{\psset{%
		linewidth=0.3pt,linestyle=solid,linecolor=black,
		dotsize=2.5pt,dotsep=2.5pt,arrowsize=4pt,
		fillstyle=none,fillcolor=white,
		showpoints=false,arrows=-,linearc=0,framearc=0,
		hatchsep=2pt,hatchwidth=0.2pt,nodesep=4pt,opacity=1}
		\psset{gridcolor=black!60, subgridcolor=black!30}
		}

\psreset
\titleformat{\section}[block]{\centering\large\bfseries\sffamily}{\thesection.}{0.5em}{}
\titleformat{\subsection}[block]{\flushleft\bfseries}{\thesubsection.}{0.5em}{}
\titleformat{\subsection}[block]{\flushleft\bfseries\sffamily}{\thesubsection.}{0.5em}{}
\titleformat{\subsubsection}[runin]{\normalsize\itshape}{\bfseries\upshape\sffamily\thesubsubsection.}{0.5em}{}[.--\:]
\renewcommand{\thesubsubsection}{\arabic{section}.\arabic{subsection}.\alph{subsubsection}}
\titlespacing{\section}{0ex}{10ex}{5ex}
\titlespacing{\subsection}{0in}{6ex}{3ex}
\titlespacing{\subsubsection}{0mm}{2ex}{0.5em}
\pretitle{\begin{center}\LARGE\bfseries\sffamily}
\posttitle{\par\end{center}\vskip 0.5em}
\preauthor{\begin{center} \large \lineskip 0.5em\begin{tabular}[t]{c}}
\postauthor{\end{tabular}\par\end{center}}
\predate{\begin{center}\small}
\postdate{\par\end{center}}
\providecommand{\abstitle}[1]{{\par\vspace*{2ex}\small\bfseries\sffamily #1}\hspace*{1ex}}
\renewenvironment{abstract}%
	{\begin{center}\begin{minipage}{0.8\linewidth}%
		\setlength{\parindent}{0.0em}\abstitle{Abstract}\small}%
	{\end{minipage}\end{center}\vfill\clearpage}
\newtheoremstyle{defn}{3ex}{3ex}{}{}{\bfseries\sffamily}{}{.5em}%
     {{\thmname{#1}\:\thmnumber{#2}}\:\thmnote{\mdseries({\small#3})\:}}%
\theoremstyle{defn}
\newtheorem{definition}{Definition}

\newtheorem{assumption}{Assumption}

\newtheoremstyle{prop}{3ex}{3ex}{\itshape}{}{\sffamily\bfseries}{}{.5em}%
     {{\thmname{#1}\:\thmnumber{#2}}\:\thmnote{\mdseries({\small#3})\:}}
\theoremstyle{prop}
\newtheorem{proposition}{Proposition}
\newtheorem{theorem}[proposition]{Theorem}
\newtheorem{lemma}[proposition]{Lemma}

\newtheoremstyle{obs}{2ex}{2ex}{}{}{\itshape}{}{.5em}%
     {{\thmname{#1}\:\thmnumber{#2}}\:\thmnote{[{\small#3}]\:}}
\theoremstyle{obs}
\newtheorem{example}{Example}
\providecommand{\cont}{\ensuremath{\:\blacktriangledown}}
\providecommand{\blank}{{\,\cdot\,}}

\DeclareMathOperator*{\argmin}{arg\,min}
\DeclareMathOperator*{\hull}{co}

\DeclareMathOperator*{\supp}{supp}
\providecommand{\Exp}[2][\!]{\mathds{E}_{#1}\left[\,#2\,\right]}

\providecommand{\Conv}[1][]{\xrightarrow[#1]{\quad}}
\providecommand{\Char}{\mathds{1}}
\providecommand{\Real}{{\mathds{R}}}
\providecommand{\Natural}{{\mathds{N}}}

\providecommand{\as}{\ensuremath{\mathrm{a.s.}}}

\newcommand{\act}{\mathbf{a}}
\newcommand{\mess}{\mathbf{m}}
\newcommand{\att}{{\Pi}}
\newcommand{\BR}{\mathrm{BR}}

\DeclareMathOperator*{\attain}{att}
\DeclareMathOperator*{\qconv}{env_{q}}
\DeclareMathOperator*{\conv}{env}

\begin{document}
\maketitle
\begin{abstract}
	I propose a cheap-talk model in which the 
	sender can use private messages and only cares about persuading a subset of her audience. 
	For example, a candidate only needs to persuade a majority of the electorate
	in order to win an election.
	I find that senders can gain credibility
	by speaking truthfully to some receivers while lying to others. 
	In general settings, the model admits information transmission in equilibrium 
	for some prior beliefs.
	The sender can approximate her preferred outcome
	when the fraction of the audience she needs to persuade is sufficiently small. 
	I characterize the sender-optimal equilibrium and the benefit
	of not having to persuade your whole audience in separable environments.
	I also analyze different applications 
	and verify that the results are robust to some perturbations of the model,
	including non-transparent motives as in \cite{CS},
	and full commitment as in \cite{KG}.
	\abstitle{Keywords} Cheap talk  $\cdot$ Information transmission $\cdot$ Persuassion 
	\abstitle{JEL classification} D83 $\cdot$ C72 $\cdot$ D72 $\cdot$ L15
\end{abstract}

A politician running for office only needs half plus one of the votes. 
A seller with a capacity constraint only needs to persuade a certain number of consumers to purchase her product. 
A person looking for a job may apply to many positions, 
but she only has to convince a single firm to extend an offer. 
This paper studies the problem of an informed sender who can engage in private
conversation with many receivers
and cares about the behavior of \emph{some but not all} of them. 
My main finding is that having to persuade only part of an audience 
significantly facilitates information transmission and increases persuasion power.  

Let us examine the first example in more detail. 
Suppose a  politician (the sender, she) is running for office. 
All voters (the receivers, he) share the same preferences.
The unknown state of the world equals either $0$ or $1$. 
Each voter will vote for the politician 
if his expectation about the state of the world is greater than $1/2$.
The voters share a common prior expectation in the interval $(1/3,1/2)$.
Suppose that the politician learns the true state of the world,
and can engage in private cheap talk with each voter via targeted ads. 
I claim that, if there electorate is large enough, 
then there exists an equilibrium in which  
she  wins the election for sure, regardless of the state.

This is possible because 
the politician only needs to persuade half plus one of the electorate in order to win the election. 
She can do so with the following strategy. 
If the state is indeed $1$, then she will let every single voter know this fact.
If the state is $0$, then she will randomly choose half plus one of the voters
and tell them that the state is $1$, despite the fact that it is not.

A voter that receives a message saying that the state equals $1$
knows that this could be a lie. 
However, he also knows that he would be more likely to receive 
this message if it was actually true. 
Hence, the message conveys some information. 
When the population is large enough, 
it conveys sufficient information to overturn prior beliefs
arbitrarily close to $1/3$.
In that case, every voter who receives this message prefers to vote for the politician.

I study a general cheap-talk model with many ex-ante homogeneous receivers and 
both public and private communication.
I depart from the literature by assuming that there are $n$ receivers,
but the sender only cares about the highest $n_0 < n$ actions taken.
In such cases, 
the utility of the sender can be completely determined by strict subsets of the receivers.
Thus, she only needs to persuade part of her audience 
in order to maximize her utility. 
I call the gap between $n$ and $n_0$ an \emph{excess audience}.

I find that the sender can influence the behavior of receivers in equilibrium 
in a very wide class of environments as long as there is an excess audience 
(Proposition \ref{prop:IT}).
In some environments, 
effective information transmission is \emph{only}
possible if there is an excess audience (Proposition \ref{prop:nece}). 
When the fraction of the audience that the sender  cares about is small enough,
she can achieve her preferred outcome in equilibrium (Proposition \ref{prop:first-best}).

I characterize the sender's benefit from having an excess audience under a separability assumption (Theorem \ref{prop:value}). 
\cite{lipnowski2017} characterize the sender's maximum equilibrium payoff 
when the sender cares about her entire audience in terms of her value function.
The value function is the highest payoff the sender can obtain when 
all the receivers behave optimally given their posterior beliefs. 
Under Lipnowski and Ravids' assumptions, 
the sender's maximum equilibrium payoff equals
the quasiconcave envelope of her value function.
I find that an additional step is needed in the presence of an excess audience.

This step involves a generalization of the 
politician's communication strategy described above. 
The sender starts by randomly and privately splitting her audience into a \emph{target} 
audience that she wants to persuade, and the rest of the receivers. 
Receivers in the target audience always receive whichever message 
induces the behavior most favorable to the sender. 
The communication strategy for the rest of the receivers
is chosen to maximize the credibility of the message sent to the target group. 
This message conveys information because individual receivers  
are not told whether they were assigned to the target audience. 

This kind of strategy allows the sender to implant 
a fixed posterior belief in a fixed proportion of the audience
regardless of the state.
I say that such beliefs are \emph{attainable}.
The set of attainable beliefs admits a simple and computationally tractable characterization
(Lemmas \ref{lemma:attain} and \ref{lemma:ver}). 
When the sender wishes  to persuade her entire audience,  the only attainable belief is the prior. 
However, the set of attainable beliefs is strictly increasing with the size of the excess audience.  
The missing step to characterize the sender's maximum equilibrium payoff
is to replace the original value 
with the maximum value over the set of attainable beliefs. 

The benefit from having an excess audience is always non-negative.
It is strictly positive for some prior beliefs.
And it  is monotone in the fraction of the audience that the sender wishes to persuade.
Moreover, as the fraction of the audience that the sender cares about converges to zero,
the set of attainable beliefs totally covers the interior of the simplex. 
Consequently, the maximum equilibrium payoff approaches the best feasible 
payoff for the sender (Proposition \ref{prop:gap}).  

%

Section \ref{sec:app} applies the techniques 
developed in the paper to analyze the election example, financial advice, 
and the role of advertisement for crowdfunding.
Section \ref{sec:ext} considers two extensions of the model. 
First, I analyze a model with full commitment as in \cite{KG}
and the information design literature \citep{BM,taneva}.
I find a characterization of the maximum sender equilibrium payoff
in the full-commitment game, assuming that the state space is finite. 
The characterization is similar to the one for the cheap-talk game. 
The only difference is that it uses the concave envelope of the value function
instead of the quasiconcave envelope. 

I also analyze an example with the classic quadratic loss functions from \cite{CS}.
This example does not satisfy all the assumptions required for my characterization. 
However, it is still possible to use  strategies with a random target audience
in order to  transmit information.
When the fraction of the audience that the sender cares about is small, 
the sender can approximate her preferred outcome 
and transmit large amounts of information to most of her audience
in equilibrium.
Unlike the case without an excess audience, 
information transmission is possible for any degree of bias. 

Since the seminal work of Vincent Crawford and Joel Sobel,
different authors have found different mechanisms 
for an expert to gain credibility. 
Information transmission is possible via cheap talk
when incentives are not too misaligned,
or there are multiple senders \citep{battaglini2002},
or multiple dimensions of information \citep{chakraborty2010},
or strategic complementarities \citep{levy2004takes,baliga2012strategy},
or the sender has transparent motives \citep{lipnowski2017},
among other reasons.
An excess audience is a novel mechanism 
which allows for information transmission in 
some settings in which none of the aforementioned mechanisms operate. 

Some authors have studied cheap-talk communication with multiple audiences.
However, this literature has focused on situations when 
the sender cares about the actions of all receivers,
either directly or indirectly. 
\cite{farrellgibbons} showed that 
senders with multiple audiences sometimes prefer public communication and sometimes 
private communication.
\cite{mariagreg} show that the sender might be strictly better off 
by combining both types of messages. 
Hence, I allow the sender to use both private and public messages. 

\cite{suraj} study the problem of an informed sender who employs
cheap talk to try to prevent an ethnic conflict. 
Their model has a large audience, and the sender is 
allowed to use private messages. 
However, they restrict attention to strategies that are anonymous,
conditional on observed heterogeneity. 
This restriction precludes the strategies with random target audiences
that I analyze. 
Instead, they exploit preference complementarities 
in order to find an equilibrium with effective information transmission.

There is a large body of literature 
using cheap talk to study information transmission 
between politicians and electorates,
dating at least as far back as \cite{harrington1992revelation}.
Some recent work in this area includes 
\cite{schnakenberg2015expert}, 
\cite{panova2017partially}, 
\cite{jeong2019using}, 
and \cite{kartik2017informative}.
Other recent papers analyze the problem 
from the information design perspective,
including  \cite{alonso2016persuading},
and \cite{chan2019pivotal}.
Within this literature, 
the papers that assume talk is cheap focus on public messages. 
My contributions highlight 
the importance of private anonymous communication (e.g., through social media).


\section{Model}
There is one sender $s$, and a set of receivers $r \in R = \{1,\ldots,{n}\}$.
The sender and receivers share a common prior belief $\pi_0$
about the true state $\theta_0\in \Theta$.
The sender learns the state.
She then sends a private message $m_r^p\in M$ to  each receiver $r$,
and a public message $m_0 \in M_0$ to all receivers.
Each receiver observes the compound message $m_r = (m_0,m_r^p)$, 
but observes neither the state nor other receivers' private messages. 
Then, all receivers observe an uninformative public sunspot $\omega^0$ 
distributed uniformly on $[0,1]$.
Finally, each receiver $r$ chooses an action $a_r \in A$.
All receivers have identical preferences.
The utility of $r$ depends only on his own action and the state.
It is given by $u_R(a_r,\theta)$.
The sender's utility $u_S(a_1,\ldots,a_n)$ does \emph{not} depend on the state.

I impose some technical restrictions.%
\footnote{%
	I adopt the following notational conventions throughout the paper. 
	For each separable metric space $Y$, 
	$\mathscr{B}_Y$ denotes the Borel $\sigma$-algebra on $Y$,
	and 
	$\Delta{Y}$ the set of probability measures on $(Y,\mathscr{B}_Y)$
	endowed with the weak* topology. 
	Given measures $\pi,\tau\in \Delta{Y}$, 
	$\pi\ll \tau$ denotes absolute continuity, $\pi\sim \tau$ denotes equivalence,
	and $d\pi/d\tau$ denotes the Radon-Nikodym derivative of $\pi$ with respect to $\tau$.
	The support of $\pi$ is denoted by $\supp \pi$.
	Given a set $X$, a function $g:X\to \Delta Y$,
	a point $x\in X$ and an event $Y'\in\mathscr{B}_Y$, let $g(Y'|x):= [g(x)](Y')$.
	Profiles of elements of $X$ are denoted by $\mathbf{x}=(x_1,\ldots,x_n)\in X^n$.
	Throughout the paper, ``a.s.'' means ``almost surely with respect to $\pi_0$''.
} 
Each of $A$, $\Theta$ and $M$ is a compact separable metric space containing at least two elements. 
$M_0$ and $M$ are rich enough as to \emph{not} restrict the set of equilibrium outcomes.%
	\footnote{A sufficient condition is $\|M\|,\|M_0\| \geq \|(A \times \Theta \times [0,1])^n\|$.} 
The utility functions are continuous. 


I study the perfect Bayesian equilibria of this game. 
\emph{Communication strategies}, \emph{receiver strategies},
and \emph{updating rules} are measurable maps 
$\mu:\Theta\to\Delta (M_0\times M^n)$,
$\alpha_r:M_0\times M\times [0,1]\to A$,
and $\beta_r:M_0\times M \to \Delta \Theta$, respectively. 
Let $\mu_r(\blank|\theta)$ denote the marginal distribution over $m_r$ induced by $\mu(\theta)$.
Let $\BR(\pi)$ be the set of actions that maximize $\int_{\Theta}  u_R(a,\theta) \, d\pi(\theta)$.
An \emph{equilibrium} is a tuple $(\bm\alpha,\bm\beta,\mu)$ consisting
a profile of {receiver strategies},
a profile of {updating rules},
and a  {communication strategy},
such that
\begin{enumerate}
\item For every receiver $r$, $\beta_r$ is consistent with Bayes' rule given $\mu_r$ and $\pi_0$.
\item $\alpha_r(m,\omega) \in A^*(\beta_r(m_r))$ for every receiver $r$ and compound message $m_r$.
\item For every message profile $\mess\in M_0\times M^n$,
	if there exists a state $\theta$ such that $\mu(\mess|\theta)>0$,
	then $\mess$ maximizes
	$ \int_0^1\! u_S(\alpha_1(m_1,\omega),\ldots,\alpha_n(m_n,\omega)) \, d\omega $.
\end{enumerate}

\subsection{Pivotal part of the audience}
Throughout most of the paper,
I maintain the assumption that the sender only cares about the tail of the empirical 
distribution of actions taken by the receivers.

\begin{assumption}\label{ass:tail}
	$A\subseteq \Real$ and there exists an integer $n'$ such that 
	for every pair of action profiles $\act$ and $\tilde\act$,
	if $a^{(i)} = \tilde{a}^{(i)}$ for every $i\geq n + 1 - n'$,
	then $u_S(\act) = u_S(\tilde\act)$,
	where $a^{(i)}$ denotes the $i$-th order statistic of $\act$.
\end{assumption}

Define the \emph{pivotal number of receivers} to be the smallest integer $n_0  \in \{0,\ldots,n\}$
that satisfies the condition from Assumption \ref{ass:tail}. 
The \emph{pivotal fraction of the audience} is $\gamma_0 = n_0/n$.
The sender's utility only depends on the highest $n_0$ actions taken by  the receivers. 
If $n_0 = 0$, then the sender is indifferent between all outcomes. 
If $n_0 = n$, then the sender cares about the entire empirical distribution of 
receiver actions, but not about the identity of the receivers taking each action. 
If $0 < n_0 <n$
%
then the sender cares about persuading fewer receivers than she can talk to.
In that case, I say that there is an \emph{excess audience}.

\begin{example}
Recall the election example from the introduction.
Letting $a_r=1$ denote a vote for the sender and $a_r=0$ a vote against her,
the outcome of the election is determined by the median action. 
If $n$ is odd, then  $n_0 = (n+1)/2$ and $\gamma_0 = (n+1)/2n$.
See Section \ref{sec:election} for more on this example. 
\end{example}

\begin{example}
Suppose that the sender is the owner of a coffee shop from a local franchise $T$. 
The shop is located at a tourist destination with several shops from the same franchise. 
Most of the potential customers are travelers who are unfamiliar with the franchise
and will buy coffee from it at most once. 
The utility of a potential customer equals $\theta_0$ 
if they buy a coffee at one of the  coffee shops ($a_r=1$),
and $0$ otherwise ($a_r=0$).
The sender's profit is normalized to equal the number of customers  she serves. 
She can talk with $n$ passing-by receivers, but she can serve at most $n_0$ of them. 
Receivers who choose $a_r=1$ but are beyond the capacity of the sender,
will buy their coffee at a different shop from the same franchise. 
Therefore, $u_S(\act) = \sum_{i=1}^{n_0} a^{(n+1-i)}$.
\end{example}	

\section{Attainable posteriors}\label{sec:posterior}
This section discusses two technical lemmas that drive the rest of the results. 
Readers interested in the main results can skip to Section \ref{sec:main}.
A key step in my analysis is to determine
the maximum influence that the sender can exert over the beliefs of part of her audience.
If the sender wants to guarantee that there are always $n_0$ receivers having certain posterior beliefs, what values can these posteriors take? 

\begin{definition}\label{def:attainable}
	For $\gamma\in[0,1]$, a belief $\pi' \in \Delta\Theta$ is $\gamma$-\emph{attained} by
	a communication strategy $\mu$ and a profile of updating rules $\bm\beta$ if
	\begin{enumerate}
	\item For every receiver $r$, $\beta_r$ is consistent with Bayes' rule given $\mu$.
	\item For every state $\theta$ and every message profile $\mess$ in the support of $\mu(\theta)$,
			there exists a set $T \subseteq R$ 
		such that $\|T\| \geq n\gamma$ and $\beta_r(m_r) = \pi$ for all $r\in T$.
	\end{enumerate}
\end{definition}

Say that $\pi$ is $\gamma$-attainable if there exist $\mu$ and $\bm\beta$ that $\gamma$-attain it.
Let $\att(\gamma)$ be the set of $\gamma$-attainable beliefs.
The following lemma characterizes the set of  $\gamma$-attainable beliefs
using a set of linear restrictions on likelihood ratios.

\begin{lemma}\label{lemma:attain}
	For all $\gamma\in(0,1]$ and $\pi \in \Delta\Theta$, the following statements are equivalent 
	\begin{enumerate}
	\item $\pi$ is $\gamma$-attainable.
	\item $\pi_0(E) \pi(E') \geq \gamma \pi_0(E')\pi(E)$ 
		for any two events $E,E'\in\mathscr{B}_\Theta$.
	\item $\pi \ll \pi_0$ and 
		$d\pi / d\pi_0  \in [\gamma c^0,c^0]$ a.s.\ for some $c^0>0$.%
		\footnote{%
		Condition (\emph{iii}) is a stronger version of
		the condition from Theorem 2.1 in \cite{diaconis1982updating}.
		They find that a posterior belief $\pi$ is consistent with a prior $\pi_0$
		and Bayes' Rule if and only if
		$\pi \ll \pi_0$ and $d\pi / d\pi_0  \leq c^0$ a.s.\ for some $c^0>0$.
		}
	\end{enumerate}
\end{lemma}

Lemma \ref{lemma:attain} implies that $\att(\gamma)$ is a nonempty closed and convex polytope.
Note that all possible posterior beliefs are $0$-attainable, and only $\pi_0$ is $1$-attainable. 
All the proofs are in the appendix. 
The following example shows  one way to reach the bounds from (\emph{ii}).

\begin{example}\label{eg:binary}
Suppose that $\Theta = \{0,1\}$ and fix some $\gamma = k/n$ with $k\in\{1,\ldots,n\}$.
Identify each $\pi \in\Delta\Theta$ with the probability $p := \pi(1)$.
The only two non-trivial events to consider are $\{1\}$ and $\{0\}$.
Hence, part (\emph{ii}) of Lemma \ref{lemma:attain} implies that a belief is $\gamma$-attainable 
if and only if
\begin{align}\label{eqn:binary-A}
	\gamma \left(\dfrac{p_0}{1-p_0} \right)
		\leq \dfrac{p}{1-p}
		\leq \dfrac{1}{\gamma} \left( \dfrac{p_0}{1-p_0}\right).
\end{align}
After some simple algebra, this implies that the set of 
$\gamma$-attainable beliefs corresponds to
	\begin{align}\label{eqn:binary-B}
		\left[\dfrac{p_0}{p_0 + (1-p_0)/\gamma_0} ,
			\dfrac{p_0}{p_0 + (1-p_0) \gamma_0} \right].
	\end{align}

The upper bound in (\ref{eqn:binary-B}) can be attained by the following communication strategy.
The sender first chooses a random target audience $T\subset R$ consisting of exactly $k$ receivers.
Each receiver $r$ only observes one of two possible (compound) messages $m_r=m'$ or  $m_r = m''$.
The sender always sends message $m'$ to all the receivers in $T$. 
Receivers not in $T$ receive message $m'$ if and only if $\theta_0=1$.
This strategy results in the conditional probabilities  $\mu_r(m'|1)=1$ and $\mu_r(m''|0)= \gamma_0$.
From Bayes' rule, the posterior belief $\beta_r(1|m')$ equals the upper bound of (\ref{eqn:binary-B}).
Since at least $k$ receive message $m'$, this belief is $\gamma$-attained. 
\end{example}

Lemma \ref{lemma:ver} below asserts that each extreme point of $\att(\gamma)$ 
corresponds to a partition of states into only two blocks.
States in one block have increased likelihoods relative to the prior,
and states in the other block have decreased likelihoods. 
This characterization makes $\att(\gamma)$ computationally tractable.
It is particularly advantageous in monotone environments 
when the sender would always prefer to increase the receivers' beliefs about the state. 


\begin{lemma}\label{lemma:ver}
	For any $\gamma\in(0,1]$,
	a belief $\pi \in\Delta \Theta$ is an extreme point of $\att(\gamma)$,
	if and only if there exists an event $E^+ \in \mathscr{B}_\Theta$ such that
	for every event $E \in \mathscr{B}_\Theta$
	\begin{align}\label{eqn:ver}
		\pi(E) = \dfrac{\pi_0(E\cap E^+) + \gamma \pi_0 (E\setminus E^+)}{\pi_0(E^+) + \gamma \pi_0(\Theta\setminus E^+)} .
	\end{align}	
\end{lemma}

\begin{example}\label{eg:three}
	Suppose that $\Theta = \{1,2,3\}$
	and $\pi_0 = (1/2,1/3,1/6)$.
	The prior likelihood ratios are 
	$\pi_0(1)/\pi_0(2) = 3/2$, $\pi_0(1)/\pi^(3) = 3$, and $\pi^(2)/\pi^(3) = 2$.
	It follows from part (\emph{ii}) of Lemma \ref{lemma:attain}
	that a posterior $\pi$ is \sfrac{1}{2}-attainable
	if and only if  $\pi(1)/\pi(2) \in [3/4,3]$, $\pi(1)/\pi(3) \in [3/2,6]$, 
	and $\pi(2)/\pi(3) \in [1,4]$.
	These conditions correspond to the cones spanning from each vertex of the 
	simplex in Figure \ref{fig:concert}.
	$\att(1/2)$ is the shaded irregular hexagon surrounding $\pi_0$.
	The vertex $\pi' = (1/3,4/9,2/9)$ is given by (\ref{eqn:ver}) with  $E^+ = \{\theta_3,\theta_2\}$.
	It maximizes $\int_\Theta \theta\,d\pi(\theta)$ subject to  $\pi \in \att(1/2)$.
\end{example}

\begin{figure}[t]
\centering
	\psset{unit=12mm}
	\begin{pspicture}(1,0.4)(10,8.6)
	\rput{15}(5,5){
		\pnode(4;360){o1}\psdot(o1) 
		\pnode(4;120){o2}\psdot(o2) 
		\pnode(4;240){o3}\psdot(o3) 
		\pspolygon(o1)(o2)(o3)
		\pnode(!\psGetNodeCenter{o2} \psGetNodeCenter{o1} 
			3 o1.x mul 1 o2.x mul add 4 div
			3 o1.y mul 1 o2.y mul add 4 div){pH12}
		\pnode(!\psGetNodeCenter{o2} \psGetNodeCenter{o1} 
			3 o1.x mul 4 o2.x mul add 7 div
			3 o1.y mul 4 o2.y mul add 7 div){pL12}
		\pnode(!\psGetNodeCenter{o3} \psGetNodeCenter{o1} 
			3 o1.x mul 2 o3.x mul add 5 div
			3 o1.y mul 2 o3.y mul add 5 div){pH13}
		\pnode(!\psGetNodeCenter{o3} \psGetNodeCenter{o1} 
			6 o1.x mul 1 o3.x mul add 7 div
			6 o1.y mul 1 o3.y mul add 7 div){pL13}
		\pnode(!\psGetNodeCenter{o2} \psGetNodeCenter{o3} 
			4 o2.x mul 1 o3.x mul add 5 div
			4 o2.y mul 1 o3.y mul add 5 div){pL23}
		\pnode(!\psGetNodeCenter{o2} \psGetNodeCenter{o3} 
			1 o2.x mul 1 o3.x mul add 2 div
			1 o2.y mul 1 o3.y mul add 2 div){pH23}
		\psset{linecolor=black!90,linestyle=dotted,fillstyle=solid,fillcolor=red!40!blue!80!black!10}
			\pspolygon(o1)(pH23)(pL23)
			\pspolygon(o2)(pH13)(pL13)
			\pspolygon(o3)(pH12)(pL12)
		\psset{fillstyle=none}
			\pspolygon(o1)(pH23)(pL23)
			\pspolygon(o2)(pH13)(pL13)
			\pspolygon(o3)(pH12)(pL12)
		\psreset%
		\pnode(!\psGetNodeCenter{o2} \psGetNodeCenter{o1} 
			3 o1.x mul 2 o2.x mul add 5 div
			3 o1.y mul 2 o2.y mul add 5 div){p12}
		\pnode(!\psGetNodeCenter{o3} \psGetNodeCenter{o1} 
			3 o1.x mul 1 o3.x mul add 4 div
			3 o1.y mul 1 o3.y mul add 4 div){p13}
		\pnode(!\psGetNodeCenter{o2} \psGetNodeCenter{o3} 
			2 o2.x mul 1 o3.x mul add 3 div
			2 o2.y mul 1 o3.y mul add 3 div){p23}
		\psset{linestyle=dotted}
			\pssavepath{pathA}{\psline(o1)(p23)}
			\pssavepath{pathB}{\psline(o2)(p13)}
			\psline(o3)(p12)
		\psreset%
		\pspolygon(o1)(o2)(o3)
		\psintersect[showpoints,name=prior]{pathA}{pathB}
		\psset{linestyle=none}
			\pssavepath{patha1}{\psline(o1)(pH23)}
			\pssavepath{pathb1}{\psline(o3)(p12)}
			\psintersect[name=ver1]{patha1}{pathb1}
			\pssavepath{patha2}{\psline(o1)(pL23)}
			\pssavepath{pathb2}{\psline(o3)(p12)}
			\psintersect[name=ver2]{patha2}{pathb2}
			\pssavepath{patha3}{\psline(o2)(pH13)}
			\pssavepath{pathb3}{\psline(o1)(p23)}
			\psintersect[name=ver3]{patha3}{pathb3}
			\pssavepath{patha4}{\psline(o2)(pL13)}
			\pssavepath{pathb4}{\psline(o1)(p23)}
			\psintersect[name=ver4]{patha4}{pathb4}
			\pssavepath{patha5}{\psline(o3)(pH12)}
			\pssavepath{pathb5}{\psline(o2)(p13)}
			\psintersect[name=ver5]{patha5}{pathb5}
			\pssavepath{patha6}{\psline(o3)(pL12)}
			\pssavepath{pathb6}{\psline(o2)(p13)}
			\psintersect[name=ver6]{patha6}{pathb6}
		\psreset%
		\pspolygon*[opacity=0.5,linecolor=red!40!blue!80!black]
			(ver11)(ver31)(ver61)(ver21)(ver41)(ver51)
		\psdots(ver31)(prior1)
		\uput{5pt}[180]{-15}(ver31){$p'$}
		\uput{5pt}[180]{-15}(prior1){$\pi_0$}
			\psset{offset=15pt,nodesep=0,arrows=|-|,tbarsize=5pt 15}
				\pcline(o2)(p12) \ncput*[nrot=:U]{$3\sqrt{2}/5$}
				\pcline(p12)(o1) \ncput*[nrot=:U]{$2\sqrt{2}/5$}
			\psreset%
	}
			\uput[015](o1){$\theta=1$}
			\uput[135](o2){$\theta=2$}
			\uput[255](o3){$\theta=3$}
	\end{pspicture}\\
	\caption{\sfrac{1}{2}-attainable beliefs for Example \ref{eg:three}.}
	\label{fig:concert}
\end{figure}
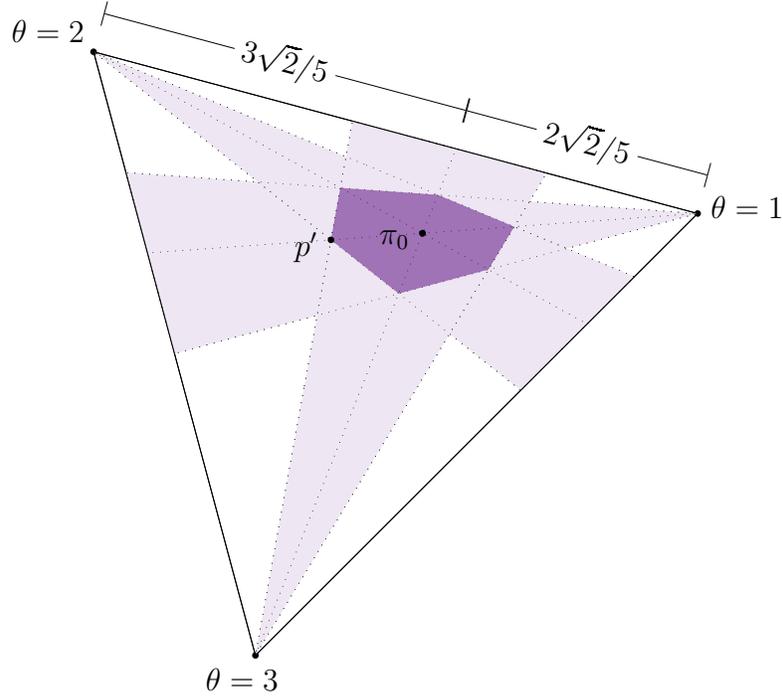

\begin{example}\label{eg:cont}
Suppose $\theta^0$ is distributed uniformly on $[0,1]$.
What is the $\gamma$-attainable belief that maximizes the expectation of the state?
Theorem 32.3 in \cite{convex} implies that the maximum is attained at an 
extreme point of $\att(\gamma)$.
Lemma \ref{lemma:ver} thus implies that the maximizer $\pi^*$ takes the form (\ref{eqn:ver}),
and it must be the case that $E^+ = [\theta_1,1]$
for some $\theta_1\in(0,1).$
It follows that
\begin{align}\label{eqn:eg-cont}
	\int_0^1 \! \theta \,d\pi^*(\theta)
		& = \int_0^{\theta_1} \frac{\gamma \theta}{\gamma \theta_1 + 1 - \theta_1} d\theta
			+  \int_{\theta_1}^1 \frac{\theta}{\gamma \theta_1 + 1 - \theta_1} d\theta 
		  = \frac{\gamma\theta_1^2 + 1-\theta_1^2}{2(\gamma \theta_1 + 1 - \theta_1)}.
\end{align} 
This expression is maximized when $\theta_1 = 1/(1+\sqrt{\gamma})$,
and the maximum is $\int_\Theta  \theta \,d\pi^*(\theta) = 1/(1+\sqrt{\gamma})$.
Note that this maximum equals $1/2$ when $\gamma=1$, equals $1$ when $\gamma=0$,
and is strictly decreasing in $\gamma$.
\end{example}

\section{Information transmission and persuasion} \label{sec:main}

\subsection{Persuading a small part of an audience} 
It is possible to exert great influence over very small fractions of an  audience. 
Suppose 
that the sender's utility is maximized by a constant action profile.
Further suppose that this action profile is a best response to a belief $\pi^*$
with a Radon-Nikodym derivative bounded both above and away from zero.
Lemma \ref{lemma:attain} implies that $\pi^*$
would be $\gamma$-attainable for sufficiently low values of $\gamma$.
Hence, when the sender only cares about small fractions of her audience,
she would be able to reach her preferred outcome in equilibrium.
Formally,

\begin{assumption}\label{ass:max}
	There exist  
	an action $a^*$ and a belief $\pi^*$ such that 
	 $a^*\in\BR(\pi^*)$
	 and 
$u_S(a^*,\ldots,a^*) \geq u_S(\act)$ for every action profile $\act$.
\end{assumption}

\begin{assumption}\label{ass:extra}
	$\pi^*\ll\pi_0$,
	and there exist two numbers $0<\ushort\nu \leq \bar\nu<\infty$
	such that $d\pi^*/d\pi_0 \in [\ushort\nu,\bar\nu]$ a.s..%
	\footnote{If $\|\Theta\|<+\infty$, 
		Assumption \ref{ass:extra} holds if and only if $\supp\pi_0=\supp\pi^*$.}
\end{assumption}




\begin{proposition}\label{prop:first-best}
	Under assumptions \ref{ass:tail}--\ref{ass:extra}
	there exists  $\bar\gamma = \bar\gamma(\pi_0,\pi^*) \in (0,1)$ 
	such that if 
	$\gamma_0\leq \bar\gamma$, 
	then the game admits an equilibrium in which the sender obtains her preferred outcome. 
\end{proposition}


\subsection{Effective information transmission}
I have assumed that the sender has \emph{transparent motives},
in that her preferences do not depend on the state and are thus common knowledge. 
Cheap-talk models with transparent motives often allow for some information transmission in equilibrium.
See, for instance, Theorem 1 in \cite{chakraborty2010} and
Proposition 1 in \cite{lipnowski2017}.
The question I address is whether the sender can transmit sufficient information
in order to influence the behavior of her audience to her benefit.

Let  $v_0$ denote the sender's \emph{value function}.
That is, $v_0(\pi)$ specifies the maximum utility that the sender could obtain 
if \emph{all} receivers shared a posterior  $\pi$ and acted optimally,
\begin{align}
	\label{eqn:vS}
	v_0(\pi) = \max\left\{ u_S(\act)\:\big|\: a_r\in\BR(\pi) \text{ for all receivers }  r \right\}.
\end{align}
Say that an equilibrium exhibits \emph{effective information transmission}
if the sender's expected equilibrium payoff is strictly greater than $v_0(\pi_0)$.
An excess audience allows for effective information transmission
for some prior beliefs under two mild sensitivity assumptions that rule out trivial cases.


\begin{assumption}\label{ass:responsive}
	$u_S(a^*,\ldots,a^*) > u_S(a',\ldots,a')$ for every action $a'\neq a^*$,
	and there exists $\pi' \in \Delta \Theta$ such that $a^*\not\in \BR(\pi')$,
	$\pi'\ll\pi^*$,
	and there exist two numbers $0<\ushort\nu'\leq \bar\nu'<\infty$
	such that $d\pi'/d\pi^* \in [\ushort\nu',\bar\nu']$ a.s..%
	\footnote{If $\|\Theta\|<+\infty$, the last condition is satisfied whenever 
		$\supp\pi'=\supp\pi^*$.}
\end{assumption}

\begin{proposition}\label{prop:IT}
	Under assumptions \ref{ass:tail}, \ref{ass:max}, and \ref{ass:responsive},
	if there is an excess audience, then there exists 
	a nonempty set  $P = P(\pi^*,\pi',\gamma_0) \subseteq \Delta\Theta$ 
	such that the game has an equilibrium with effective information transmission
	as long as $\pi_0 \in P$.
	Moreover, if $\|\Theta\|<+\infty$, then $P$ has a nonempty interior.%
\end{proposition}


\subsection{Sender-optimal equilibrium under separability}

The sender's \emph{maximum equilibrium payoff} $v^*$ is  
the maximum utility that the sender can obtain in any equilibrium.
This section characterizes $v^*$ under the following separability assumption.

\begin{assumption}\label{ass:sep}
There exist a strictly increasing function $U_S:A\to\Real$ such that
\begin{align}\label{eqn:sep}
	u_S(\act) = \dfrac{1}{n_0}\sum_{i=1}^{n_0} U_S\left( a^{(n+1-i)} \right).
\end{align}
\end{assumption}

The characterization relies on two operators 
defined on the set of upper-semicontinuous functions from beliefs to sender payoffs. 
First,  $\qconv v$ is the quasiconcave envelope of $v$.
That is, $\qconv v$ is the pointwise-minimum, 
quasiconcave and upper semicontinuous function that majorizes $v$.
Second, $\attain v$ gives the maximum of $v$ arising from $\gamma_0$-attainable beliefs,
i.e.,
\begin{align}
	\label{eqn:vS-hat}
	\attain v(\pi) = \max\left\{  v(\pi') \:\big|\: \pi\in \att(\gamma_0,\pi)\right\},
\end{align}
where $\att(\gamma_0,\pi)$ is the set of beliefs that would be $\gamma_0$-attainable 
if $\pi_0 = \pi$.
Intuitively, $\qconv v$ operates by ``flooding the valleys'' 
while $\attain v$ is obtained by ``widening the hills.''
See the left and center panels of Figure \ref{fig:bookie} in Section \ref{sec:bookie}
for an example. 
%
%

%

\begin{theorem}\label{prop:value}
	Under assumptions \ref{ass:tail}  and \ref{ass:sep}, $v^* = \attain \qconv v_0(\pi_0)$.
\end{theorem}
%

Lemma $1$ implies that $\attain \qconv v_0 = \qconv v_0$ when $\gamma_0=1$.
Hence, Theorem \ref{prop:value} reduces to Theorem 2 in \cite{lipnowski2017}
in that case.
However, the two results differ whenever there is an excess audience and 
the assumptions from Proposition \ref{prop:IT} hold. 
The difference between the results corresponds to the benefit from an excess
audience defined in the following subsection. 

\subsection{The benefit from an excess audience and private communication}
What happens to the sender's payoff when she has to persuade
a larger or smaller fraction of her audience?
Proposition \ref{prop:nece} below  gives sufficient conditions
under which an excess audience is \emph{necessary} for effective information transmission. 
These conditions are satisfied by the election, excess capacity, and labor market 
applications in Section \ref{sec:app}.

\begin{proposition}
	\label{prop:nece}
	Under assumptions \ref{ass:tail}  and \ref{ass:sep}, 
	if $v_0$ is quasi-concave and there is an equilibrium with effective information transmission, 
	then there is an excess audience. 
\end{proposition}

Define the \emph{benefit from excess audience} to be the 
difference between $v^*$ and the maximum equilibrium payoff to the sender 
in an alternative environment with $n_0=n$.
%
It follows from Theorem \ref{prop:value} and Lemma \ref{lemma:attain} that
this benefit equals the difference between $v^*$ and $\qconv v_0(\pi_0)$.
Moreover, Assumption \ref{ass:sep} guarantees that both $v^*$ and $\qconv v_0(\pi_0)$ are measured in the same units.

\begin{proposition}\label{prop:gap}
	Under assumptions \ref{ass:tail}--\ref{ass:sep},
	there exist $\bar{\gamma}=\bar{\gamma}(\pi_0,\pi^*) > 0 $ 
	and a nonempty set $P=P(\pi^*,\pi',\gamma_0)\subseteq \Delta\Theta$
	such that the benefit from excess audience
	\begin{enumerate}
	\item is non-negative and non-increasing in $\gamma_0$,
	\item is strictly positive whenever $\pi_0 \in P$ and $\gamma_0<1$, and
	\item equals  $\tilde{u}_S(a^*) - v_0(\pi_0)$ whenever $\gamma_0 \leq \bar\gamma$.
	\end{enumerate}
\end{proposition}

Consider an alternative model in which the sender can use  only public messages. 
She still cares only about the actions of part of her audience.
However, public messages only allow her to persuade either all of the receivers or none of them.
The \emph{benefit from private communication}
is the gap between $v^*$ and the maximum sender equilibrium value 
in this alternative model. 
Without private messages, the only $\gamma_0$-attainable belief is the prior.
Hence, the benefit from private communication coincides with the benefit from excess audience
under Assumption \ref{ass:sep}.

\section{Examples}
\label{sec:app}

\subsection{Targeted political campaigns}
\label{sec:election}
	Social media allows politicians to personalize campaign advertising at a low cost. 
	Suppose that the sender is a politician running for office.
	Each receiver $r$ will either vote for the sender ($a_r=1$) or against her ($a_r=0$).
	The state is either $0$ or $1$,
	and voters share a common prior belief with $p_0 := \pi_0(1) \in(0,1)$.
	The sender knows the state and engages in private cheap talk with each receiver 
	via targeted online ads. 
	
	The sender wins the election if she obtains a super-majority 
	of at least $\gamma\in(0,1)$ of the votes.
	Her utility is $1$ if she wins the election, and $0$ otherwise. 
	The pivotal number of receivers is $n_0 = \min\{ n' \:|\: n' > n\gamma \}$.
	All receivers share the same preferences.
	Receiver $r$ prefers $a_r=1$ if and only if his posterior beliefs satisfy
	$p_r := \pi_r(1) \geq \eta$, where $\eta \in(0,1)$ is a fixed parameter.

	When is victory attainable for the politician?
	She wins the election when the posterior beliefs of at least
	$n_0$ voters satisfy $\pi_r > \eta$.
	From (\ref{eqn:binary-A}),
	there exists a $\gamma_0$-attainable belief satisfying this condition
	if and only if
	\begin{align}\label{eqn:election-attainable}
	\dfrac{\eta}{1-\eta} < \dfrac{1}{\gamma_0} \left(\dfrac{p_0}{1-p_0}\right).
	\end{align}
	In such cases, there exists an equilibrium in which the sender
	always wins the election regardless of the state. 
	The condition is satisfied whenever:
	(\emph{i}) the voters prior attitude towards the sender is positive ($p_0$ is high), 
	(\emph{ii}) the voters have a low bar for the sender ($\eta$ is small), 
	or (\emph{iii}) the sender only needs a small fraction of votes in order to win the election
	($\gamma_0$ is small). 
	For the case $\gamma =1/2$ and $\eta = 1/2$, 
	(\ref{eqn:election-attainable}) reduces to the condition  $\pi_0 > 1/3$
	from the introduction.

\subsection{Financial advice}
\label{sec:bookie}
	The state $\theta_0 \in\{0,1\}$ indicates the winner of a rigged boxing match. 
	The sender is an informed bookie who knows the state
	and would like to persuade the receivers to place large bets.
	However, she is time constrained.
	She can talk with $n$ receivers, but she can handle at most $n_0\leq n$ bets.
	The rest of the bets will be handled by other bookies. 
	The total utility of the sender equals  $V_0 + \eta V_1$, 
	where $V_\theta$ is the total volume of bets on $\theta$ that she handles,
	and $\eta > 1$ is a fixed parameter. 

	Each receiver starts with the same initial wealth $w>0$ and places a bet $a_r\in [-w,w]$.
	A positive bet represents a bet on $\theta_0=1$ 
	whereas a negative bet represents a bet on $\theta_0=0$.
	Bets on different states have different exogenous net returns $\rho_0,\rho_1>0$,
	with $\rho_0\rho_1 < 1$.
	Receivers have logarithmic Bernoulli utility functions.
	For example, a receiver with beliefs $\pi$ that places a bet on $\theta_0 =1$ maximizes 
	\begin{align}
		\int\limits_\Theta \! u_R(a_r,\theta) \,d\pi(\theta)
			= p \log(w+a_r) + (1-p) \log(w-a_r),
	\end{align}
	subject to $a_r\in(0,w]$, where $p = \pi(1)$.
	
	This example deviates slightly from our environment
	because the sender cares about both tails of the distribution of actions.
	However, the conclusion of Theorem \ref{prop:value} still applies.
	The receivers' best response correspondence is given by%
	\begin{align}
		\BR(\pi) = \left\{\begin{array}{ll}
			\dfrac{w}{\rho_1}[(1+\rho_1)p - 1] 
				& \quad\text{if}\quad p \geq \dfrac{1}{1+\rho_1}\\[2ex]
			-\dfrac{w}{\rho_0}[(1+\rho_0)(1-p) - 1] 
				& \quad\text{if}\quad p \leq \dfrac{\rho_0}{1+\rho_0}\\[2ex]
			0 & \quad\text{otherwise}
		\end{array}\right..
	\end{align}	
	The function $\attain v_0$ can be computed by substituting the bounds from (\ref{eqn:binary-B}) into $u_S(\BR(p))$.
	Figure \ref{fig:bookie} illustrates $\qconv v_0$ (left), $\attain v_0$
	 for $\gamma_0 \in\{1/3, 1/10\}$ (middle),
	and  $v^*$ for $\gamma_0 =1/3$ (right).
	The gap between the envelopes in the left and right panels corresponds to 
	both the benefit from private communication
	and the benefit from having more receivers than capacity to take bets.
	This example can also be interpreted as the problem of a time-constrained financial adviser 
	who takes prices as given and sells multiple negatively correlated instruments. 

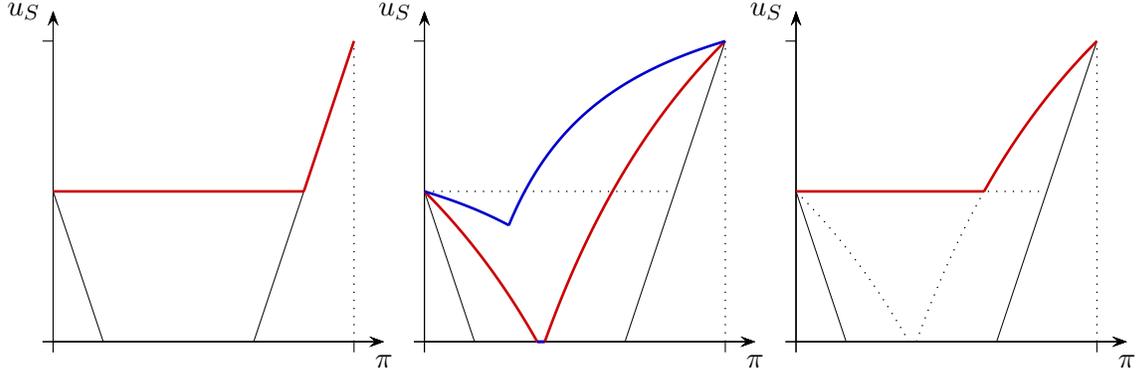
\begin{figure}[t]
	\centering
	\psset{unit=4mm}
	\begin{pspicture}(0,-1)(12,11)
	\small%
	\psaxes[labels=none,tickstyle=bottom,Dx=10,Dy=10]{->}(0,0)(0,0)(11,11)[$\pi$,270][$u_S$,180]
		\psplot{20 3 div}{10}{ 3 x mul   20  sub} 
		\psplot{0}{10 6 div}{ 10 x 6 mul  sub 2 div} 
	\psset{linestyle=dotted,linewidth=0.6pt}
		\psline(10,0)(10,10)
	\psset{linewidth=1pt,linecolor=red!80!black,linestyle=solid}
		\psline(0,5)(! 50 6 div 5)
		\psplot{50 6 div}{10}{ 3 x mul   20  sub} 
	\end{pspicture}
	\begin{pspicture}(0,-1)(12,11)
	\small%
	\psaxes[labels=none,tickstyle=bottom,Dx=10,Dy=10]{->}(0,0)(0,0)(11,11)[$\pi$,270][$u_S$,180]
		\psplot{20 3 div}{10}{ 3 x mul   20  sub} 
		\psplot{0}{10 6 div}{ 10 x 6 mul  sub 2 div} 
	\psset{linestyle=dotted,linewidth=0.6pt}
		\psline(10,0)(10,10)
		\psline(0,5)(! 50 6 div 5)
	\psset{linewidth=1pt,linecolor=red!80!black,linestyle=solid}
		\psplot{0}{30 8 div}{ 10 10 x mul 30 2 x mul sub div  6 mul  sub 2 div} 
		\psplot{20 5 div}{10}{ 3 x 30 mul 2 x mul 10 add div mul   20  sub} 
		\psline(! 30 8 div 0)(! 20 5 div 0)
	\psset{linecolor=blue!80!black}
		\psplot{0}{2.8}{ 10 10 x mul 100 9 x mul sub div  6 mul  sub 2 div} 
		\psplot{2.8}{10}{ 3 x 100 mul 9 x mul 10 add div mul   20  sub} 
		\psline(! 30 8 div 0)(! 20 5 div 0)
	\end{pspicture}
	\begin{pspicture}(0,-1)(11,11)
	\small%
	\psaxes[labels=none,tickstyle=bottom,Dx=10,Dy=10]{->}(0,0)(0,0)(11,11)[$\pi$,270][$u_S$,180]
		\psplot{20 3 div}{10}{ 3 x mul   20  sub} 
		\psplot{0}{10 6 div}{ 10 x 6 mul  sub 2 div} 
	\psset{linestyle=dotted,linewidth=0.6pt}
		\psline(10,0)(10,10)
		\psline(0,5)(! 50 6 div 5)
		\psplot{0}{30 8 div}{ 10 10 x mul 30 2 x mul sub div  6 mul  sub 2 div} 
		\psplot{20 5 div}{10}{ 3 x 30 mul 2 x mul 10 add div mul   20  sub} 
	\psset{linewidth=1pt,linecolor=red!80!black,linestyle=solid}
		\psline(0,5)(! 50 8 div 5)
		\psplot{50 8 div}{10}{ 3 x 30 mul 2 x mul 10 add div mul   20  sub} 
	\end{pspicture}
	\caption{Value functions for the bookie example with $w=0.5$,
		$\eta^2 = 2$, $\rho_1 = 1/2$, and $\rho_0 = 1/5$.
		Left panel: $\qconv v_0$. 
		Center panel: $\attain v_0$ with $\gamma_0=1/3$ (red) and $\gamma_0 = 1/10$ (blue).
		Right panel: $v^*$ with $\gamma_0 = 1/3$.
		}
	\label{fig:bookie}
\end{figure}

\subsection{Crowdfunding}
The sender owns a start-up company financed via an online crowdfunding platform.
The receivers are potential backers. 
Each receiver pledges an investment level $a_r \geq 0$.
Say that the company is \emph{backed} if $\sum_{r\in R} a_r \geq \eta$,
where $\eta>0$ is a fixed parameter.
The sender gets a payoff of $1$ if the company is backed and a payoff of $0$
otherwise. 

A backed company might be a success or a failure. 
The company succeeds with probability  $\theta_0\in[0,1]$ if it is backed,
and it fails for sure otherwise.
The prior belief about $\theta_0$ is uniform on $[0,1]$.
The sender is informed about $\theta_0$
and can communicate with the receivers via private cheap talk. 

If the company is not backed, the pledged investments are refunded. 
If the company is backed and fails, all the investments are lost. 
Otherwise, the investors receive a net return $\rho > 0$.
The receivers have logarithmic Bernoulli utility functions.
Hence, 
\begin{align}
	u_{R}(a_r,\theta) 
		& = -\theta e^{-(w+\rho a_r)} - (1-\theta)e^{-(w-a_r)}
\end{align}
if the company is backed, and $u_{R}(a_r,\theta) = - e^{-w}$  otherwise.%
Lemmas \ref{lemma:attain} and \ref{lemma:ver} can be used 
to construct equilibria in which the sender benefits from having an excess audience. 
Note that the company will be backed if at least $n'$ receivers to pledge at least $\eta/n'$. 
If  $\eta/n' < 1$, a receiver who is optimistic enough about $\theta_0$
would be willing to pledge that amount. 
Optimistic beliefs can be $({n'}/{n})$-attained as long as $n$ is sufficiently larger than $n'$.

\begin{proposition}\label{prop:kickstarter}
	If $n \geq 12 \eta / w \rho^2$,
	then the crowdfunding game has a perfect Ba\-ye\-sian equilibrium 
	in which the company is backed for sure regardless of the state.
\end{proposition}

%

\section{Extensions}
\label{sec:ext}

\subsection{Information design}
Suppose now that the  sender chooses and commits to a communication strategy 
\emph{before} learning the state of nature. 
This timing corresponds to the information design paradigm used by \cite{KG}.
See also \cite{BM}.
Define a \emph{commitment protocol} to be  a tuple $(\bm\alpha,\bm\beta,\mu)$ satisfying 
	conditions (\emph{i}) and (\emph{ii}) in the definition of an equilibrium.
The sender's \emph{maximum commitment payoff} $v^{**}(p)$ is  
the maximum utility that the sender can obtain in any commitment protocol.

Since every equilibrium is a commitment protocol, Propositions \ref{prop:first-best} and \ref{prop:IT} continue to hold in the game with commitment. 
Also, it is possible to obtain a geometric characterzation of $v^{**}$.
Let $\conv v_0$ denote the concave envelope of $v_0$,
that is, the pointwise-minimum, concave function that majorizes $v_0$.

\begin{theorem}\label{prop:commitment}
	If $\|\Theta\|<+\infty$ and Assumption \ref{ass:sep} holds,
	then $v^{**} = \attain \conv v_0(\pi_0)$.%
	\footnote{%
		The proof is based on Proposition 1 in \cite{KG}, which assumes a finite state space. 
		The discussion in Section 3 of their online appendix suggests 
		that it might be possible to extend the result to compact separable state spaces. }
\end{theorem}

Figure \ref{fig:election} illustrates this result for the election example
with $\gamma_0=\eta=1/2$.
The figure on the left panel shows $\conv v_0$,
which corresponds to the maximum commitment value without an excess audience \citep{KG}.
This would also be the maximum commitment value for the sender if she was restricted 
to use only public messages.
The middle panel shows $v^* = \attain v_0$.
The right panel shows $v^{**}$.
The benefit from commitment is thus given by the gap between the right and the middle panel. 
The benefit from an excess audience under commitment is given by the gap between the right and the left panel.

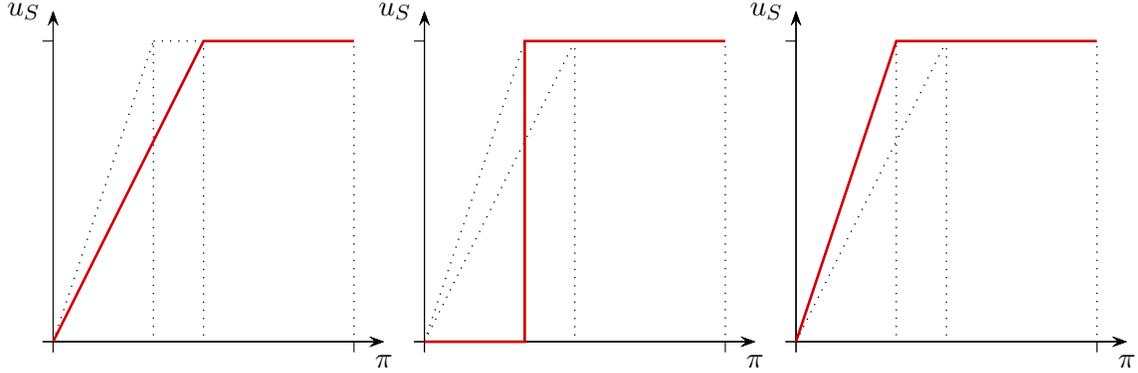
\begin{figure}[t]
	\centering
	\psset{unit=4mm}
	\begin{pspicture}(0,-1)(12,11)
	\small%
	\psaxes[labels=none,tickstyle=bottom,Dx=10,Dy=10]{->}(0,0)(0,0)(11,11)[$\pi$,270][$u_S$,180]
	\psset{linestyle=dotted,linewidth=0.6pt}
		\psline(! 10 3 div 10)(10,10)(10,0)
		\psline(0,0)(! 10 2 div 10)(! 10 2 div 0)
		\psline(0,0)(! 10 3 div 10)(! 10 3 div 0)
	\psset{linewidth=1pt,linecolor=red!80!black,linestyle=solid}
		\psline(0,0)(! 10 2 div 10)(10,10)		
	\end{pspicture}
	\begin{pspicture}(0,-1)(12,11)
	\small%
	\psaxes[labels=none,tickstyle=bottom,Dx=10,Dy=10]{->}(0,0)(0,0)(11,11)[$\pi$,270][$u_S$,180]
	\psset{linestyle=dotted,linewidth=0.6pt}
		\psline(! 10 3 div 10)(10,10)(10,0)
		\psline(0,0)(! 10 2 div 10)(! 10 2 div 0)
		\psline(0,0)(! 10 3 div 10)(! 10 3 div 0)
	\psset{linewidth=1pt,linecolor=red!80!black,linestyle=solid}
		\psline(0,0)(! 10 3 div 0) (! 10 3 div 10)(10,10)
	\end{pspicture}
	\begin{pspicture}(0,-1)(11,11)
	\small%
	\psaxes[labels=none,tickstyle=bottom,Dx=10,Dy=10]{->}(0,0)(0,0)(11,11)[$\pi$,270][$u_S$,180]
	\psset{linestyle=dotted,linewidth=0.6pt}
		\psline(! 10 3 div 10)(10,10)(10,0)
		\psline(0,0)(! 10 2 div 10)(! 10 2 div 0)
		\psline(0,0)(! 10 3 div 10)(! 10 3 div 0)
	\psset{linewidth=1pt,linecolor=red!80!black,linestyle=solid}
		\psline(0,0) (! 10 3 div 10)(10,10)
	\end{pspicture}
	\caption{Commitment value for election example with $\gamma_0=\eta=1/2.$}
	\label{fig:election}
\end{figure}

%
%
%
%

\subsection{Lack of transparent motives}
Assuming that the sender does not care about the state 
simplifies the analysis and plays a crucial role in the proof Theorem \ref{prop:value}.
However, an excess audience can still be beneficial even
when the sender's motives are not transparent. 
Consider the classic quadratic-loss game from \cite{CS},
with the twist that the sender faces an excess audience. 

The state is distributed uniformly on $[0,1]$.
The receivers take actions in $A=[0,1]$,
and their utility is $u_R(a_r,\theta_0) = -(\theta_0-a_r)^2$.
The sender's utility is given by 
\begin{align}\label{eqn:CSu}
	u_S(\act,\theta_0) = -\dfrac{1}{n_0}\sum_{i=1}^{n_0} \left( \theta_0 + b - a^{(n+1-i)} \right)^2,
\end{align}
where $b>1/4$ is a fixed parameter measuring the \emph{bias} of the sender
relative to the receivers. 
Suppose that the sender is only allowed to use private messages. 

When $n_0$ is equal to the total number of receivers,
this is a particular instance of the environment studied by \cite{mariagreg}.
Since all the receivers are biased in the same direction,
there are only babbling equilibira. 
In contrast, when $n$ is much larger than $n_0$, 
there can be effective information transmission.

\begin{proposition}\label{prop:CS}
	For all $\epsilon > 0$, 
	there exists $\ushort{n}<\infty$ such that
	whenever $n \geq \ushort{n}$
	the maximum sender equilibrium value 
	in the quadratic-loss game is greater than $-\epsilon$
	if $\theta_0+b \leq 1$,
	and greater than $-b^2-\epsilon$ otherwise. 
\end{proposition}

The proof is constructive. 
The sender uses a strategy based on a finite partition of $[0,1]$.
She randomly splits the receivers into two groups of sizes 
$n-n_0$ and $n_0$.
She then reveals truthfully which block of the partition contains 
$\theta_0$ to the members of the first group. 
She misleads the members of the second group 
so that they choose her preferred action.
When $n$ is very large, it is possible to construct 
incentive compatible equilibria of this sort with very fine partitions. 

\section{Closing remarks}
When talk is cheap, 
information transmission requires the 
sender to be indifferent between all messages she uses. 
The present work identifies a novel mechanism that can create indifference. 
When the sender only cares about persuading a strict subset of her audience, 
she is indifferent between the messages she sends to the rest of the receivers. 
It is possible for her to gain credibility by being truthful 
with some receivers while lying to others. 
This mechanism can greatly facilitate information transmission 
and increase the sender's power to persuade.

The present work provides a full characterization of the
sender-optimal equilibrium assuming that 
the receivers do not care about each other actions, 
the sender has transparent motives, 
and her preferences are monotone and satisfy a separability condition. 
These restrictions greatly simplify the analysis.
They make it possible to characterize the set of equilibria 
combining the idea of $\gamma_0$-attainability
with the techniques from \cite{lipnowski2017}.
However, they appear to be inessential for many of the results. 
The value of having a large audience in general settings is left as an open problem. 

\bibliography{references}
\appendix
\section{Proofs}
\label{sec:proofs}

\subsection{Attainable posteriors}
\begin{proof}[Proof of Lemma \ref{lemma:attain}]
	If $\gamma =1$, then the only belief that satisfies either (\emph{i}), (\emph{ii})
	or (\emph{iii}) is $\pi_0$, and it satisfies all three conditions. 
	Hence, for the rest of the proof, suppose that $\gamma \in (0,1)$.

((\emph{i}) $\Rightarrow$ (\emph{ii}))
	Suppose that $\pi$ is $\gamma$-attained by some $\bm\beta$ and $\mu$, 
	and fix any two events $E'$ and $E''$.
	For each receiver $r$, let $M_r$ be the set of compound messages $m$ such that $\beta_r(m) = \pi$,
	and let $\chi_r$ be the indicator that $m_r\in M_r$.
	Also let $\chi = \sum_{r\in R} \chi_r$ be the number of receivers whose posterior equals $\pi$.
	The prior beliefs $\pi_0$, communication strategy $\mu$, and updating rule $\bm\beta$
	induce a joint probability measure $\Pr$ over $\theta_0$, $m_r$, $\chi_r$ and $\chi$.

	First, suppose that  $\pi_0(E')\pi_0(E'') = 0$. 
	Since $\gamma>0$, there exists at least one receiver $r$ such that $\Pr(M_r)>0$.
	This implies that $\pi$ is obtained from $\pi_0$ using Bayes' rule.
	Therefore, $\pi\ll \pi_0$ and (\emph{iii}) holds with equality.
		
	Now consider the case $\pi_0(E')\pi_0(E'') > 0$.
	Note that
	\begin{align}\label{eqn:proof-lemm-A}
		\Exp{\chi|E'}
			= \sum_{r\in R} \Exp{\chi_r|E'} 
			=  \sum_{r\in R} \Pr(M_r|E'),
	\end{align}
	where I am using short notation for the events $\theta_0 \in E'$
	and $m_r\in M_r$.
	Since $p'$ is $\gamma$-attained by $x$ and $\bm\beta$, it follows that 
	\begin{align}\label{eqn:proof-lemm-B}
		\Pr(\chi\geq \gamma n|E') = 1 
		\qquad\Rightarrow\qquad \Exp{\chi|E'}\geq \gamma n.
	\end{align}		
	Combining (\ref{eqn:proof-lemm-A}) and (\ref{eqn:proof-lemm-B}) it follows that
	\begin{align}
		\dfrac{1}{n} \sum_{r\in R} \Pr(M_r|E') 
		\geq \gamma.
	\end{align}
	Therefore there exists receiver $r'$
	such that $\Pr(M_{r'}|E')  \geq \gamma$
	and $\Pr(M_{r'})>0$.

	It follows from Bayes' rule that
	\begin{align}
		\pi(E') 
			= \dfrac{\pi_0(E') \Pr(M_{r'}|E')}{\Pr(M_{r'}) }
		\quad\text{and}\quad
		\pi(E'') 
			= \dfrac{\pi_0(E'') \Pr(M_{r'}|E'')}{\Pr( M_{r'}) }.
	\end{align}
	Taking the ratio of these equations yields:
	\begin{align}
		\dfrac{\pi(E')}{\pi(E'')}
			= \dfrac{\pi_0(E')}{\pi_0(E'')} \cdot \dfrac{\Pr(M_{r'}|E')}{\Pr(M_{r'}|E'')}
			\geq  \dfrac{\pi_0(E')}{\pi_0(E'')}\cdot\dfrac{\gamma}{1},
	\end{align}
	which is equivalent to (\emph{ii}) after rearranging terms. 

((\emph{ii}) $\Rightarrow$ (\emph{iii}))
	Now suppose that (\emph{ii}) and take any event $E$ such that $\pi_0(E)=0$,
	and thus $\pi_0(\Omega\setminus E) = 1$.
	It follows that 
		$\gamma \pi(E) \pi_0(\Omega\setminus E) \leq \pi(\Omega\setminus E) \pi_0(E) = 0$.
	Since $\gamma>0$, this implies that $\pi(E)=0$.
	Since $E$ was arbitrary, this implies that $\pi \ll \pi_0$.
	Hence, $d\pi/d\pi_0:\Theta\to[0,\infty)$ exists by the Radon-Nikodym theorem. 
	
	Let $\bar\nu,\ushort\nu\in \Real_+ \cup\{\infty\}$ denote the supremum and infimum 
	of the support of $dp'/d\pi_0 (\theta_0)$, respectively, i.e., 
	\begin{align}\label{eqn:supRN}
		\bar\nu = \sup\left\{ \nu \:\big|\: 
			\pi_0\left( \dfrac{d\pi}{d\pi_0}(\theta) \geq \nu \right) > 0 \right\} 
	\end{align}
	and
	\begin{align}\label{eqn:infRN}
		\ushort\nu = \inf\left\{ \nu \:\big|\: 
			\pi_0\left( \dfrac{d\pi}{d\pi_0}(\theta) \leq \nu \right) > 0 \right\}.
	\end{align}
	By construction, $ 0\leq \ushort\nu < \infty$, $\ushort\nu\leq \bar\nu$,
	and $d\pi / d\pi_0  \in [\ushort\nu,\bar\nu]$ a.s..
	If $\bar\nu = \ushort\nu$, then (\emph{iii}) is satisfied by $c_0 = \bar\nu$.
	Otherwise, there exists $\nu_1,\nu_2\in \Real$ such that 
	$\ushort\nu < \nu_1 < \nu_2 < \bar\nu$. 
	I will show that if (\emph{ii}) holds, 
	then  $\gamma \nu_2 \leq \nu_1$ for any such pair.
	Therefore, $\gamma \bar\nu \leq \ushort\nu$
	and (\emph{iii}) is satisfied by $c_0 = \bar\nu$.

	Let $E_1 = \{\theta \:|\: d\pi/d\pi_0 (\theta) \leq \nu_1\}$ 
	and $E_2 = \{\theta \:|\: d\pi/d\pi_0 (\theta) \geq \nu_2\}$.
	Since $d\pi/d\pi_0$ is measurable, so are $E_1$ and $E_2$.
	The way $\bar\nu$ and $\ushort\nu$ were defined implies that $\pi_0(E_1)>0$ and $\pi_0(E_2)>0$. 
	It follows from the Radon-Nikodym theorem  that
	\begin{align}\label{eqn:proof-lemm-C}
		\pi(E_1)
			& = \int\limits_{E_1} \dfrac{d\pi}{d\pi_0}(\theta) \,d\pi_0(\theta) 
			  \leq \int\limits_{E_1} \nu_1 \,d\pi_0(\theta|E_1) = \nu_1 \pi_0(E_1).
	\end{align}
	By a similar argument, it follows that 
	\begin{align}\label{eqn:proof-lemm-D}
		\pi(E_2) \geq \nu_2 \pi_0(E_2).
	\end{align}
	The fact that $\nu_2 < \bar\nu$ implies that  $\pi_0(E_2)>0$. 
	Hence, we can divide (\ref{eqn:proof-lemm-C}) by (\ref{eqn:proof-lemm-D}).
	Doing so yields
	\begin{align}\label{eqn:proof-lemm-D}
		\dfrac{\nu_1}{\nu_2} \dfrac{\pi_0(E_1)}{\pi_0(E_2)} 
			\geq \dfrac{\pi(E_1)}{\pi(E_2)} \geq \gamma\dfrac{\pi_0(E_1)}{\pi_0(E_2)},
	\end{align}
	where the second inequality follows from condition (\emph{ii}). 
	Therefore, $\gamma \nu_2 \leq \nu_1$.
	

((\emph{iii}) $\Rightarrow$ (\emph{i}))
	The proof is constructive.
	Suppose that condition (\emph{iii}) holds  and let 
	 $\phi:\Theta\to\Real$ be the function given by
	\begin{align}\label{eqn:proof-lemm-eta}
		\phi(\theta) = \dfrac{1}{c^0} \cdot  \dfrac{d\pi}{d\pi_0}(\theta),
	\end{align}
	Also, let $k = \min\{n'\in\Natural \:|\: n'\geq \gamma n \}$.
	Consider the following communication strategy.
	The sender always sends the same non-informative public message $m_0 = \emptyset$.
	She chooses a random target audience $T \subseteq R$ consisting of exactly $k$ receivers 
	uniformly from $\{R'\subseteq R \:|\: \|R'\| = k \}$.
	She sends the (compound) message $m' = (\emptyset,m^1)$ 
	with probability $1$ to all receivers in $T$ regardless of the state.
	For each state $\theta\in\Theta^*$,
	each receiver $r$ \emph{not} in $T$ receives message $m'$ with probability 
	\begin{align}
		\mu(m'|\theta,r\not\in T) =  \dfrac{1}{1-\gamma}\left( \phi(\theta) - \gamma \right),
	\end{align}
	and with the remaining probability he receives a different fixed compound 
	message $m'' = (\emptyset,m^2)$.%
	\footnote{Note that the sender only uses two private messages, 
		and that public messages are not informative. 
		These facts play no role in this proof, 
		but they are used in the proof of other propositions 
		which rely on on this communication strategy.} 
	
	Condition (\emph{iii}) implies that there exists 
	$E^*\in\mathscr{B}_\Theta$ such that $\pi_0(E^*)=1$ 
	and  $\phi(\theta) \in [\gamma,1]$ for every $\theta\in E^*$.
	Hence, $\mu_r(m'|\theta,r\not\in T) \in [0,1]$ for every $\theta\in E^*$.
	Therefore, $\mu$ is a well defined communication strategy 
	(up to a null event, in which it can be redefined arbitrarily). 

	A receiver $r$ who receives message $m'$ does not know whether he belongs to $T$ or not.
	He thus updates based on the probability 
	\begin{align}
		\mu_r(m'|\theta) = (1-\gamma) \mu_r(m'|\theta,r\not\in T) + \gamma = \phi(\theta)
	\end{align}
%
	Now, take any event $E\subseteq E^*$ such that $\pi(E) > 0$ and note that
	\begin{align}
		 \mu_r(m'|E) 
		 	& = \int\limits_E \! \mu_r(m'|\theta) \,d\pi_0(\theta|E)  
		 	  = \int\limits_E \! \phi(\theta) \cdot \dfrac{1}{\pi_0(E)}\, d\pi_0(\theta) \nonumber\\
		 	& = \dfrac{1}{c^0 \pi_0(E)} \int\limits_E  \dfrac{d\pi}{d\pi_0}(\theta) \, d\pi_0(\theta)
		 	  = \dfrac{1}{c^0 \pi_0(E)} \cdot \pi(E),
	\end{align}
	where the second equality follows because $\pi\ll \pi_0$, and thus $\pi_0(E) > 0$;
	the third one follows from (\ref{eqn:proof-lemm-eta}),
	and the last one follows from the Radon-Nikodim theorem. 
	Therefore, using Bayes' rule yields
	\begin{align}
		\beta_r(E|m') 
			& = \dfrac{\pi_0(E) \mu_r(m'|E)}{\mu_r(m')} 
			  = \dfrac{1}{c^0 \mu_r(m')} \cdot \pi(E) \propto \pi(E).
	\end{align}
	Since $E$ was arbitrary and $\pi_0(E^*)=1$, 
	it follows that $\beta_r(m') = \pi$ a.s..
	Since there are always at least $n\gamma$ receivers with $m_r=m'$,
	we can conclude that $\pi$ is $\gamma$-attained by $\bm\beta$ and $\mu$.
\end{proof}		

The following lemma is an intermediate step to prove Lemma \ref{lemma:ver}.

\begin{lemma}\label{lemma:inter}
	Given any $\gamma \in (0,1)$,
	a $\gamma$-attainable belief $\pi \in \att(\gamma)$ is an extreme point of $\att(\gamma)$
	if and only if there exists $c^0>0$ such that
	$d\pi / d\pi_0  \in \{\gamma c^0,c^0\}$ a.s..
\end{lemma}
\begin{proof}{}
($\Leftarrow$)
	Suppose that there exists $c^0>0$ such that
	$d\pi / d\pi_0  \in \{\gamma c^0,c^0\}$ a.s..
	Lemma \ref{lemma:attain} implies that $\pi_0$ is $\gamma$-attainable.
	Suppose towards a contradiction that there exist 
	$\pi',\pi''\in\att(\gamma)\setminus\{\pi\}$ and  $\lambda\in(0,1)$
	such that $\pi = \lambda \pi' + (1-\lambda) \pi''$.

	Since $\pi\neq \pi'$, it follows that $\pi_0(d\pi/d\pi_0(\theta) \neq d\pi'/d\pi_0(\theta)) > 0$.
	Without loss of generality, suppose that $\pi_0(dp'/d\pi_0(\theta) > c^0) > 0$.
	Lemma \ref{lemma:attain} then implies that $\pi_0(d\pi'/d\pi_0(\theta) \leq \gamma c^0) = 0$.
	Since $\pi \in\hull\{\pi',\pi''\}$,
	it follows that $d\pi''/d\pi_0(\theta) < \gamma c^0$ a.s.\ 
	on $E^-:= \{\theta \:|\: d\pi/d\pi_0(\theta) = \gamma c^0\}$.
	Lemma \ref{lemma:attain} then implies that $\pi_0(d\pi''/d\pi_0(\theta) \geq c^0) = 0$.
	In particular, 
	$d\pi''/d\pi_0(\theta) < c^0$ a.s.\ 
	on $E^+:= \{\theta \:|\: d\pi/d\pi_0(\theta) =  c^0\}$.
	Hence, $d\pi''/d\pi_0 < d\pi/d\pi_0$ a.s..
	However, this implies the following contradiction:
	\begin{align}
		\pi''(\Theta) 
			= \int\limits_\Theta \dfrac{d\pi''}{d\pi_0}(\theta)\,d\pi_0(\theta)
			< \int\limits_\Theta \dfrac{d\pi}{d\pi_0}(\theta)\,d\pi_0(\theta)
			= \pi(\Theta) = 1 \quad\cont.
	\end{align}
	Hence, we can conclude that $\pi$ is an extreme point of $\att(\gamma)$.

($\Rightarrow$)
	Fix a belief $\pi \in \att(\gamma)$.
	Let $\bar\nu,\ushort\nu\in \Real_+ \cup\{\infty\}$ be the supremum and infimum 
	of the support of $d\pi/d\pi_0 (\theta_0)$, as defined in (\ref{eqn:supRN}) and (\ref{eqn:infRN}).
	The proof proceeds in two steps. 
	First, I will show that if $\pi_0(d\pi / d\pi_0(\theta) \in \{\ushort\nu,\bar\nu\}) < 1$,
	then $\pi$ is \emph{not} an extreme point of $\att(\gamma)$.
	Second, I will show that if $\ushort\nu\not = \gamma \bar\nu$, 
	then $\pi$ is \emph{not} an extreme point of $\att(\gamma)$.
	
\emph{Step 1}---Suppose that $\pi_0(d\pi / d\pi_0(\theta) \in \{\ushort\nu,\bar\nu\}) < 1$.
	I will show that there exist $\pi',\pi'' \in \att(\gamma) \setminus\{\pi\}$
	such that $\pi = 0.5 \pi' + 0.5 \pi''$.
	Consequently, $\pi$ is \emph{not} an extreme point of $\att(\gamma)$.

	Since $\pi_0(d\pi / d\pi_0(\theta) \in [\ushort\nu,\bar\nu]) = 1$
	and $\pi_0(d\pi / d\pi_0(\theta) \in \{\ushort\nu,\bar\nu\}) < 1$,
	there exist numbers $\nu_1,\nu_2 \in (\ushort\nu,\bar\nu)$ such that 
	$\nu_1\leq \nu_2$ and $\pi_0(E^0) >0$, where 
	$E^0 := \left\{ \theta\in\Theta \:|\: {d\pi}/{d\pi_0}(\theta) \in [\nu_1,\nu_2] \right\}$.
	Moreover, it follows from the definition of $\bar\nu$ and $\ushort\nu$
	that $\pi_0(E^0)<1$.
	Since $\pi\ll \pi_0$, it follows that $\pi(\Theta\setminus E^0) > 0$.
	Fix any number $\epsilon>0$ such that 
	\begin{align}\label{eqn:lemm-epsilon}
		\epsilon < \min\left\{ 1,\:
			\dfrac{\nu_1}{c^1},\: 
			\dfrac{\bar\nu - \nu_2}{\bar\nu + c^1},\:
			\dfrac{\nu_1 - \ushort\nu}{\ushort\nu + c^1}
		\right\},
	\end{align}
	where 
	\begin{align}
		c^1 = \dfrac{\pi(\Theta\setminus E^0)}{\pi_0(E^0)} \in (0,\infty).
	\end{align}

	Let $\pi'$ and $\pi''$ be the beliefs given by 
	\begin{align}\label{eqn:lemma-extremes}
		\pi'(E) = \int\limits_E  \dfrac{d\pi'}{d\pi_0}(\theta)\, d\pi_0(\theta)
		\quad\text{and}\quad
		\pi''(E) = \int\limits_E  \dfrac{d\pi''}{d\pi_0}(\theta)\, d\pi_0(\theta),
	\end{align}
	where $d\pi'/d\pi_0 :\Theta\to \Real_+$ and $d\pi''/d\pi_0 :\Theta\to \Real_+$ are given by
	\begin{align}
		 \dfrac{d\pi'}{d\pi_0}(\theta) =\left\{\begin{array}{ll}
		 	\dfrac{d\pi}{d\pi_0}(\theta) + \epsilon c^1
		 		& \quad\text{if}\quad \theta\in E^0\\[2ex]
		 	(1-\epsilon)\dfrac{d\pi}{d\pi_0}(\theta) 
		 		& \quad\text{if}\quad \theta\not\in E^0
		 \end{array}\right.,
	\end{align}
	and 
	\begin{align}
		 \dfrac{d\pi''}{d\pi_0}(\theta) =\left\{\begin{array}{ll}
		 	\dfrac{d\pi}{d\pi_0}(\theta) - \epsilon c^1
		 		& \quad\text{if}\quad \theta\in E^0\\[2ex]
		 	(1+\epsilon)\dfrac{d\pi}{d\pi_0}(\theta) 
		 		& \quad\text{if}\quad \theta\not\in E^0
		 \end{array}\right..
	\end{align}
	
	It is straightforward to verify that $d\pi/d\pi_0 = 0.5 d\pi'/d\pi_0 + 0.5 d\pi''/d\pi_0$
	and, consequently, $\pi = 0.5 \pi' + 0.5 \pi''$.
	The first two bounds from (\ref{eqn:lemm-epsilon}) imply 
	that $d\pi'/d\pi_0\geq 0$ and $d\pi''/d\pi_0 \geq 0$ a.s..
	Also, note that
	\begin{align}
		\pi'(\Theta) 
			& = \int\limits_{E^0} 
					\left[\dfrac{d\pi}{d\pi_0}(\theta) 
					+ \dfrac{\epsilon \pi(\Theta\setminus E^0)}{\pi_0(E^0)}\right] \, d\pi_0(\theta)
				+ \int\limits_{\Theta\setminus E^0} 
					(1-\epsilon)\dfrac{d\pi}{d\pi_0}(\theta) \,d\pi_0(\theta) \nonumber\\
			& = \pi(E^0) 
					+ \dfrac{\epsilon \pi(\Theta\setminus E^0)}{\pi_0(E^0)} \cdot \pi_0(E^0)
				+ (1-\epsilon) \pi(\Theta\setminus E^0) = \pi(\Theta) = 1,
	\end{align}
	and, by a similar argument $\pi''(\Theta) = 1$.
	It follows that $\pi'$ and $\pi''$ are well defined beliefs and are absolutely
	continuous with respect to $\pi_0$.
	
	From the third bound in (\ref{eqn:lemm-epsilon}) it follows that
	\begin{align}
		\epsilon < \dfrac{\bar\nu - \nu_2}{\bar\nu + c^1}
			& \quad\Rightarrow\quad \bar\nu\epsilon + c^1\epsilon < \bar\nu - \nu_2 \nonumber \\
			& \quad\Rightarrow\quad \dfrac{d\pi}{d\pi_0}(\theta) + \epsilon c^1 
				\leq \nu_2 + c^1\epsilon < (1-\epsilon)\bar\nu.
	\end{align}
	Therefore,	$d\pi' / d\pi_0  \in [(1-\epsilon)\ushort\nu,(1-\epsilon)\bar\nu]$
	a.s..
	By a similar argument, the fourth bound in (\ref{eqn:lemm-epsilon})
	implies that $d\pi'' / d\pi_0  \in [(1+\epsilon)\ushort\nu,(1+\epsilon)\bar\nu]$
	a.s..
	Lemma \ref{lemma:attain} implies that $\ushort\nu \geq \gamma \bar\nu$
	and thus $\pi'$ and $\pi''$ are $\gamma$-attainable.
	
\emph{Step 2}---
	Now suppose that 
	$d\pi / d\pi_0  \in \{\ushort\nu,\bar\nu\}$ a.s.,
	but $\ushort\nu\neq \gamma\nu$.
	Again, 
	I will show that there exist $\pi',\pi'' \in \att(\gamma) \setminus\{\pi\}$
	such that $\pi = 0.5 \pi' + 0.5 \pi''$.	
	Since the argument is very similar to the one used in the previous step, 
	I will omit some details. 

	From Lemma \ref{lemma:attain}, it follows that $\ushort\nu \geq \gamma\bar\nu$.
	Hence, it must be the case that $\ushort\nu > \gamma\bar\nu > 0$.
	Since $\gamma\leq 1$, this implies that $\ushort\nu\neq \bar\nu$.
	Consider the events $E^+ =\{\theta \:|\: d\pi/d\pi_0(\theta) = \bar\nu \}$ 
	and $E^-=\{\theta \:|\: d\pi/d\pi_0(\theta) = \ushort\nu \}$.
	We know that $\pi_0(E^+) + \pi_0(E^-)=1$.
	Moreover, $\pi_0(E^+) = 1$ or $\pi_0(E^-) = 1$ would only be possible 
	if $\bar\nu = \ushort\nu = 1$.
	Therefore, $\pi_0(E^+)\pi_0(E^-) \neq 0$.
	Fix any $\epsilon>0$ such that 
	\begin{align}\label{eqn:lemm-epsilon2}
		\epsilon < \min\left\{ 
			\ushort\nu,\:
			\dfrac{\bar\nu}{c^2},\:
			\dfrac{\bar\nu - \ushort\nu}{1 + c^2},\:
			\dfrac{\ushort\nu - \gamma\bar\nu}{1 + \gamma c^2}
		\right\},
	\end{align}
	where $c^2 = {\pi_0(E^-)}/{\pi_0(E^+)} \in (0,\infty)$.	

	Let $\pi'$ and $\pi''$ be defined as in (\ref{eqn:lemma-extremes}),
	but with $d\pi'/d\pi_0$ and $d\pi''/d\pi_0$ given by
	\begin{align}
		 \dfrac{d\pi'}{d\pi_0}(\theta) =\left\{\begin{array}{ll}
		 	\bar\nu+ \epsilon c^2
		 		& \quad\text{if}\quad \theta\in E^+\\[2ex]
		 	\ushort\nu - \epsilon
		 		& \quad\text{if}\quad \theta\in E^-
		 \end{array}\right.,
	\end{align}
	and
	\begin{align}
		 \dfrac{d\pi''}{d\pi_0}(\theta) =\left\{\begin{array}{ll}
		 	\bar\nu -  \epsilon c^2
		 		& \quad\text{if}\quad \theta\in E^+\\[2ex]
		 	\ushort\nu + \epsilon
		 		& \quad\text{if}\quad \theta\in E^-
		 \end{array}\right..
	\end{align}
	Using an analogous argument to the one used in Step 1,
	it is possible to verify that $\pi = 0.5 \pi' + 0.5 \pi''$
	and $\pi'$ and $\pi''$ are well defined beliefs and $\pi',\pi''\in \att(\gamma)$.
\end{proof}

\begin{proof}[Proof or Lemma \ref{lemma:ver}]{}
($\Rightarrow$)
	Let $\pi$ be an extreme point of $\att(\gamma)$.
	From Lemma \ref{lemma:inter}, there exists a constant $c^0$ such that
	$d\pi / d\pi_0  \in \{\gamma c^0,c^0\}$ a.s..
	Let $E^+ = \{\theta\in\Theta \:|\: d\pi/d\pi_0(\theta) = c^0\}$.
	We have $d\pi/d\pi_0(\theta) = \gamma c^0$ a.s.\ 
	on $\Theta \setminus E^+$.
	Hence, for every event $E$,
	\begin{align}\label{eqn:proof-vert-A}
		\pi(E) 
			& = \int\limits_{E\cap E^+} \dfrac{d\pi}{d\pi_0} (\theta)\,d\pi_0(\theta)
				+ \int\limits_{E \setminus E^+} \dfrac{d\pi}{d\pi_0} (\theta)\,d\pi_0(\theta)
				\nonumber\\
			& = c^0  \big[ \pi_0(E\cap E^+) + \gamma \pi_0(E\setminus E^+) \big].
	\end{align}
	In particular, for $E = \Theta$,
	\begin{align}
		\pi(\Theta) = c^0  \big[ \pi_0(\Theta\cap E^+) + \gamma \pi_0(\Theta\setminus E^+) \big].
	\end{align}
	Since $\pi(\Theta) = 1$,
	it follows that $c^0 = 1/[\pi_0(\Theta\cap \Theta^+) + \gamma \pi_0(\Theta\setminus \Theta^+)]$.

($\Leftarrow$)
	Suppose that there exists a measurable set  $E^+ \subseteq \Theta$ such that
	$\pi$ satisfies (\ref{eqn:ver}) for every event $E \in \mathscr{B}_\Theta$.
	I claim that the Radon-Nikodym derivative of $\pi$ with respect to $\pi_0$
	is given by 
	\begin{align}
		\dfrac{d\pi}{d\pi_0}(\theta) = c^0
		\big[ {\Char(\theta\in E^+) + \gamma\cdot \Char(\theta\not\in E^+)}\big],
	\end{align}
	where $c^0 = 1/[\pi_0(E^+) + \gamma \pi_0(\Theta\setminus E^+)]$.
	Indeed, note that for every event $E$,
	\begin{align}
		\int\limits_E \dfrac{d\pi}{d\pi_0}(\theta)\,d\pi_0(\theta) 
			& =  \int\limits_{E\cap E^+} \!\!c^0 \,d\pi_0(\theta) 
					+ \int\limits_{ E\setminus E^+}\!\! c^0 \gamma \,d\pi_0(\theta)  \nonumber\\
			& = \dfrac{\pi_0(E\cap E^+) + \gamma \pi_0(E\setminus E^+)}{\pi_0(E^+) + \gamma \pi_0(\Theta\setminus E^+)}.
	\end{align}
	Since $\gamma c^0 \leq d\pi/d\pi_0(\theta) \leq c^0$ for all $\theta\in \Theta$,
	 the result follows from Lemma \ref{lemma:attain}.
\end{proof}

\subsection{Effective communication and persuasion}

\begin{proof}[Proof of Proposition \ref{prop:first-best}]
	The threshold is given by $\bar\gamma = \ushort\nu / \bar\nu  \in (0,1)$.
	If $\gamma_0 < \bar\gamma$, then 
	$d\pi^*/d\pi_0 \in [\gamma_0 \bar\nu, \bar\nu]$ a.s..
	Lemma \ref{lemma:attain} thus implies that $\pi^{*}$ is $\gamma_0$-attainable
	by some $\bm\beta$ and $\mu$.
	Consider any strategy profile $\bm\alpha$ such that all receivers choose best responses 
	and, in particular,  $\alpha_r(m_r,\omega) = a^*$ whenever $\beta_r(m_r) = \pi^{*}$.
	The tupple $(\bm\alpha,\bm\beta,\mu)$ constitutes an equilibrium.
	Let $\act$ be any action profile that results with positive probability in this equilibrium.
	By construction, at least $\gamma_0$ of the receivers satisfy $a_r = a^*$.
	Hence, it maximizes the sender's utility. 
\end{proof}

\begin{proof}[Proof of Proposition \ref{prop:IT}]
	Let $P^* = \BR^{-1}(a^*)$ be the set of beliefs in 
	$\Delta\Theta$ for which  $a^*$ is a best response. 
	Also, for each $\lambda \in[0,1]$ let $\pi_\lambda = \lambda\pi' + (1-\lambda)\pi^*$.
	Let $\Lambda^* = \{\lambda\in[0,1] \:|\: \pi_\lambda\in P^*\}$.
	And let $\bar\lambda = \sup \Lambda^*$ and $\bar\pi = \pi_{\bar\lambda}$.

	$\bar\lambda^*$ is well defined an belongs to $[0,1]$ because $0\in \Lambda^*$,
	and thus $\Lambda^*\neq\emptyset$.
	Take any sequence $(\lambda_k)$ in $\Lambda^*$ such that $\lambda_k\Conv \bar\lambda$.
	For any event $E$, 
	$\pi_{\lambda_k}(E) = (1-\lambda_k) \pi'(E) + \lambda_k \pi^*(E)
	\Conv (1- \bar\lambda) \pi'(E) + \bar\lambda\pi^*(E) = \bar\pi$.
	Hence $\pi_{\lambda_k}$ converges with respect to the weak topology to $\bar\pi$
	and, consequently, with also with respect to the weak* topology. 
	Since $u_R$ is continuous, $P^*$ is closed.
	Since $\pi_{\lambda_k}(E) \in P^*$ for all $k$, $\bar\pi \in \Lambda^*$. 
	Moreover, since $P^*$ is convex, so is $\Lambda^*$.
	Therefore, we can write $\Lambda^* = [0,\bar\pi]$ with $\bar\pi\in[0,1)$.
	
	Assumption \ref{ass:responsive} implies that $\pi'\sim \pi^*$ and,
	consequently, $\pi_\lambda\sim \pi_{\lambda'}$ for all $\lambda,\lambda'\in[0,1]$.
	Therefore, for every $\lambda\in(\bar\pi,1]$ we have that
	\begin{align}\label{eqn:proof-IT-A}
		\dfrac{d\bar\pi}{d\pi_\lambda} 
			& = \dfrac{d\bar\pi}{d\pi^*} \cdot \dfrac{d\pi^*}{d\pi_\lambda}
			  = \dfrac{ {d\bar\pi}/{d\pi^*}}{{d\pi_\lambda}/{d\pi^*}}
			  = \dfrac{\bar\lambda {d\pi'}/{d\pi^*} + (1-\bar\lambda)}
			  		{\lambda {d\pi'}/{d\pi^*} + (1-\lambda)},
	\end{align}	
	a.s.,
	where 	the first and second equalities follow from the chain rule,
	and the third one from the linearity of Radon-Nikodym derivatives. 
	Define $f_\lambda:\Real_+\to\Real$ by
	\begin{align}
		f_\lambda(x) = \dfrac{\bar\lambda x  + 1 - \bar\lambda}{\lambda x + 1-\lambda}.
	\end{align}
	Notice that
	\begin{align}
		\dfrac{df_\lambda}{dx}(x) = \dfrac{\bar\lambda - \lambda}{(\lambda x + 1-\lambda)^2} < 0.
	\end{align}
	Hence $f$ is decreasing, and 
	using the fact that ${d\pi'}/{d\pi^*}\in[\ushort\nu', \bar\nu'] $ a.s.,
	it follows from (\ref{eqn:proof-IT-A}) that 
	${d\bar\pi}/d\pi_\lambda\in[f_\lambda(\bar\nu'), f_\lambda(\ushort\nu')] $ a.s.. 
	
	Since there is an excess audience, 
	there exists  $\epsilon>0$ such that $\epsilon < (1-\gamma_0)/(1+\gamma_0)$.
	Note that 
	\begin{align}
		\lim_{\lambda\to\bar\lambda^+} f_\lambda(\ushort\nu') 
		= \lim_{\lambda\to\bar\lambda^+} f_\lambda(\bar\nu')= 1.
	\end{align}
	Hence, there exists a number $\delta > 0$ such that if $\lambda \in (\bar\lambda, \bar\lambda + \delta)$,
	then $f_\lambda(\ushort\nu'), f_\lambda(\bar\nu') \in (1-\epsilon,1+\epsilon)$.
	This implies that $f_\lambda(\bar\nu')/f_\lambda(\ushort\nu') \geq (1-\epsilon)/(1+\epsilon) > \gamma_0$,
	and thus,
	${d\bar\pi}/d\pi_\lambda\in [\gamma f_\lambda(\ushort\nu'), f_\lambda(\ushort\nu')]$ a.s..	
	Lemma \ref{lemma:attain} thus implies that if $\pi_0 = \pi_\lambda$, 
	then $\bar\pi$ is $\gamma_0$-attainable.
	
	Fix some $\lambda\in (\bar\lambda, \bar\lambda + \delta)$,
	and suppose that $\pi_0 = \pi_\lambda$.
	Since $\pi_\lambda\not\in P^*$, Assumption \ref{ass:responsive} implies that 
	$v_0(\pi_0) < u_S(a^*,\ldots,a^*)$.
	The sender payoff equals $u_S(a^*,\ldots,a^*)$ in the equilibrium constructed in the proof of Proposition  \ref{prop:first-best}.
	Hence, there is effective information transmission in this equilibrium. 
	Therefore, $\pi_\lambda \in P \neq \emptyset$.

	It remains to show that if $\|\Theta\|<\infty$, then $P$ has a nonempty interior.
	For that purpose, 
	let $\Pi(\gamma_0,\pi)$ denote the set of beliefs that
	would be $\gamma_0$-attainable if $pi^0=\pi$.
	From Lemma \ref{lemma:attain}, it follows that
	\begin{align}
		\pi \in \Pi(\gamma_0,\bar\pi) 
			&\ \Leftrightarrow\  
				\left( \pi\ll\bar\pi 
				\text{ and } \dfrac{d\pi}{d\bar\pi} \in [\gamma_0c^0,c^0] \as \right) \nonumber\\
			&\ \Leftrightarrow\ 
				\left( \bar\pi\ll\pi \text{ and } 
				\dfrac{d\bar\pi}{d\pi} \in \left[\gamma_0\dfrac{1}{c^0},\dfrac{1}{c^0}\right] \as \right)
			\ \Leftrightarrow\  \bar\pi \in \Pi(\gamma_0,\pi).
		\end{align} 
	
	When $\gamma_0<1$ and $\|\Theta\|<+\infty$, 
	$\bar\pi$ belongs to the interior of $\Pi(\gamma_0,\bar\pi)$.
	Therefore, there exists an open neighborhood $P'\subseteq \Pi(\gamma_0,\bar\pi)$ of $\bar\pi$
	such that $\bar\pi\in \Pi(\gamma_0,\pi)$ for all $\pi\in P'$.
	Since $\bar\pi$ belongs to the boundary of $P^*$, 
	it follows that $P'' = P'\setminus P^*\neq \emptyset$.
	Since $P^*$ is closed, $P''$ is open. 
	And, from the same argument we used for $\pi_\lambda$,
	it follows that 
	 for every $\pi \in P''$ there exists an equilibrium with effective information transmission.
	 Hence, $P''\subseteq P$.
\end{proof}

%

\subsection{Sender-optimal equilibrium}

\begin{lemma}\label{lemma:symmetric}
	Every equilibrium value can be attained by
	a symmetric equilibrium $(\bm\alpha,\bm\beta,\mu)$  with 
	 $\alpha_r = \alpha_{r'}$, $\beta_r = \beta_{r'}$, and $\mu_r = \mu_{r'}$
	for all receivers $r$ and $r'$.
\end{lemma}
\begin{proof}
	Suppose that an equilibrium value $u_S^*$ is  generated by some equilibrium 
	 $(\bm\alpha,\bm\beta,\mu)$.
	Consider the alternative strategies $\tilde\mu$ and $\tilde{\bm\alpha}$
	obtained by shuffling identities as follows. 
	First, the sender draws a  message profile $\mess\in M_0\times M^n$ 
	using the original distribution $\mu$, but does not deliver them. 
	Then, she shuffles the identity of the receivers 
	by drawing a permutation uniformly from $\{ I:R\to R \:|\: I \text{ is biyective}\}$.
	She tells each receiver which function $\alpha_{I(r)}(m_{I(r)},\blank)$ 
	they would have used in the original equilibrium
	with their swapped identity and message.
		
	Note that $\alpha_{I(r)}(m_{I(r)},\blank)$ would be a best response 
	for $r$ had he been told the shuffling $I$ and the message $m_{I(r)}$.
	The sure thing principle then implies that
	it is also a best response when this information is garbled. 
	The new strategy profile thus induces an equilibrium that yields $u_S^*$,
	and is symmetric. 
\end{proof}

\begin{lemma}[\cite{lipnowski2017}]
	\label{lemma:LR}
	Under assumption \ref{ass:sep}, $v^* \geq \qconv v_0(\pi_0)$.
\end{lemma}
\begin{proof}
	Let $u_S^* = \qconv v_0(\pi_0)$.
	Consider the alternative environment
	with $\tilde{n}=1$ and $\tilde{u}_S(a) = u_s(a,\ldots,a)$,
	and $\Theta$, $A$, $\pi_0$, and $u_R$ unchanged.
	Assumption \ref{ass:sep} implies that the value function
	of the alternative environment coincides with the value 
	function of the original environment.
	Since there is only one receiver, and the sender's utility does not depend on the state, 
	this alternative environment satisfies the assumptions in \cite{lipnowski2017}. 
	Hence, by their Theorem 2,
	there exists an equilibrium $(\tilde\alpha_1,\tilde\beta_1,\tilde\mu_1)$
	which achieves $u_S^*$.

	Consider the replica of this equilibrium given by
	$\alpha_r(m_0,m_r^0,\omega) = \tilde\alpha_1(m_0,\omega)$,
		$\beta_r = \tilde{\beta}_1$,
	and 
	$\mu(m,\emptyset,\ldots,\emptyset|\theta) = \tilde{\mu}_1(m|\theta)$, 
	where $\emptyset$ denotes a fixed non-informative private message. 
	Note that this replica uses correlated strategies which guarantee
	that all agents receive the same message and take the same action with probability $1$.
%
%
	It is straightforward to verify that 
	$(\bm\alpha,\bm\beta,\mu)$ is an equilibrium of the original environment and achieves $u_S^*$.
\end{proof}

\begin{proof}[Proof of Theorem \ref{prop:value}]
\emph{Step 1} ($\attain \qconv v_0$ is well defined)---%
Since $\Theta$ and $A$ are compact and $u_R$ is continuous, 
$v_0$ is well defined and upper-semicontinuous. 
By Lemma 5 in \cite{lipnowski2017},
$\qconv v_0$ is also well defined and upper-semicontinuous. 
From Lemma \ref{lemma:attain}, $\att$ is closed.
Prokhorov's theorem thus implies that $\att$ is compact.   
Hence, $\attain \qconv v_0$ is well defined by Weierstrass' extreme-value theorem.

\emph{Step 2} ($v^* \geq \attain \qconv v_0(\pi_0)$)---%
Let  $u_S^* = \attain \qconv v_0(\pi_0)$.
There exists some $\hat\pi\in \att$ such that $u_S^* = \qconv v_0(\hat\pi)$.
	Consider the alternative environment
	with $\tilde{n}=1$, $\tilde{u}_S(a) = U_S(a)$, $\tilde{\pi}^0 = \hat\pi$, 
	and $\Theta$, $A$, and $u_R$  unchanged.
	Assumption \ref{ass:sep} implies that the value function
	of the alternative environment coincides with $v_0$.
	Since there is only one receiver, and the sender's utility does not depend on the state, 
	this alternative environment satisfies the assumptions in \cite{lipnowski2017}. 
	Hence by their Theorem 2, there exists an equilibrium $(\tilde\alpha_1,\tilde\beta_1,\tilde\mu_1)$
	which achieves $\qconv v_0(\hat{\pi}) = u_S^*$.
	Moreover, we must have 
	\begin{align}\label{eqn:geq}
		\int_0^1 \! U_S\left( \tilde\alpha_1(m_1)  \right) \,d\omega = u_S^*,
	\end{align}
	for every (compound) message $m_1$ such that $\tilde\mu(m_1)>0$.
	Let 
	$(\hat{\bm\alpha},\hat{\bm\beta},\hat\mu)$ be the replica of 
	$(\tilde\alpha_1,\tilde\beta_1,\tilde\mu_1)$ constructed as in the proof of Lemma \ref{lemma:symmetric}.
	It would be an equilibrium of the original environment if $\pi_0 = \hat\pi$.

Let $\bar{x}$ and $\bar{\bm\beta}$ be the communication strategy and updating rule
that $\gamma_0$-attain $\hat\pi$ from the proof of Lemma \ref{lemma:attain}.
Consider the tuple $(\bm\alpha,\bm\beta,\mu)$ described as follows.
The sender first draws (but does not deliver) 
messages  $\hat{m}_0$ using $\hat{x}$,
and $\bar{\mess}\in\{m',m''\}^n$ using $\bar{x}$. 
She only sends non-informative public messages $m_0=\emptyset$. 
If $\bar{m}_r = m'$, then $r$ receives the \emph{private} message $m_r^p = (m',\hat{m}_0)$.
Otherwise, he receives the private message $m_r^p = (m'')$.
$\bm\beta$ is derived from $\mu$ using Bayes' rule. 
Actions are given by 
$\alpha_r(\emptyset,(m',\hat{m}_0),\omega) = \tilde{\alpha}_1(\hat{m}_0,\omega)$,
and $\alpha_r(\emptyset,m'',\omega)=a''$ with $a'' = \min \BR(\beta_r(m''))$.

By construction, we have that $\beta(\emptyset,(m',\hat{m}_0)) = \tilde{\beta}_1(\hat{m}_0)$.
Therefore, the strategy $\alpha_r(\emptyset,(m',\hat{m}_0),\omega)$ is a best response
for the senders. 
Since $(\hat{\bm\alpha},\hat{\bm\beta},\hat\mu)$ would be an equilibrium if $\pi_0=\hat\pi$, 
the sender cannot benefit from manipulating $\hat{m}_0$.
Hence, if $u^*_S \geq U_S(a'')$,
then $(\bm\alpha,\bm\beta,\mu)$ is an equilibrium.
Since there are always $n_0$ players who receive message $m'$
and their actions lead to $u_S^*$ (because of (\ref{eqn:geq})),
this would imply $v^*\geq \attain \qconv v_0(\pi_0)$.

Otherwise, $u^*_S < \tilde{u}_S(a'')$.
In this case, note that $\pi_0$ is a convex combination of $\hat{\pi}$
and $\beta_r(\emptyset,m'')$.
And, in turn $\hat{p} \in \hull\{ \beta(\emptyset,(m',\hat{m}_0)) \:|\: 
	\tilde{\mu}(\hat{m_0}) > 0 \}$.
Note that $v_0(\beta_r(\emptyset,m'') \geq  U_S(a'') > u^*_S$,
and $v_0(\beta_r(\emptyset,(m',\hat{m}_0)) \geq u_S^*$ (because of (\ref{eqn:geq}).
Hence, $\qconv v_0(\pi_0) \geq u_S^* = \attain \qconv v_0(\pi_0)$,
and the desired inequality $v^*\geq \attain \qconv v_0(\pi_0)$
follows from Lemma \ref{lemma:LR}.

\emph{Step 3} ($v^* \leq \attain \qconv v_0(\pi_0)$)---%
Let $u_S^*$ be an arbitrary equilibrium value. 
From Lemma \ref{lemma:symmetric}, there exists a symmetric equilibrium 
$(\bm\alpha,\bm\beta,\mu)$ that generates $u_S^*$.
	Let $\mess$ be a message profile such that $\mu(\mess)>0$. 
	Under assumption \ref{ass:sep}, there must exist a set $R(\mess)$
	with $\|R(\mess)\| \geq n_0$ and such that 
	$\sup\{ U_S(\alpha_r(m_r,\omega)) \:|\: {\omega\in[0,1]} \} \geq u_S^*$
	for every $r\in R(\mess)$.
	Let $P^0$ be the set corresponding set of posterior beliefs,
	\begin{align}\label{eqn:lemma-leq-A}
		P^0 = \left\{ \beta_r(m_r) \:\big|\quad \mu(\mess)>0 \text{ and } r\in R(\mess) \right\}.
	\end{align}
	It follows that 
	\begin{align}\label{eqn:lemma-leq-B}
		u_S^* \leq \inf\left\{ v_0(\pi) \:\big|\: \pi \in P^0 \right\}.
	\end{align}

Consider the alternative communication strategy $\mu'$ with only two 
(compound) messages $m'$ and $m''$ described as follows. 
The sender first draws a profile  $\mess$ according to $\mu$ (but does \emph{not} deliver it). 
Receiver $r$ receives message $m'$ if and only if $r\in R(\mess)$.
Since $(\bm\alpha,\bm\beta,\mu)$ is  symmetric, $\bar{\pi} := \beta_r(m')$ does not depend on $r$.
The martingale property of Bayes' rule implies that  $\bar\pi \in \hull(P^0)$.
Since there are always at least $n_0$ receivers in $R(\mess)$,
it follows that $\bar\pi$ is $\gamma_0$-attainable. 
Therefore 
	\begin{align}
		u_S^* 
			& \leq \inf_{\pi\in P^0} v_0(\pi)
			  \leq \inf_{\pi\in P^0} \qconv v_0(\pi)
			  \leq \inf_{\pi\in \hull(P^0)} \qconv v_0(\pi) \nonumber \\
			&  \leq  \qconv v(\bar{p}) \leq \attain \qconv v_0(\pi_0).
	\end{align}
The first inequality is just (\ref{eqn:lemma-leq-B}).
The second inequality follows because $\qconv v_0$ majorizes $v_0$.
The third one because $\qconv v_0$ is quasiconcave.
The fourth one because infimums are lower bounds.
The last one from the fact that $\bar{p}$ is $\gamma_0$-attainable. 	
Since, $u_S^*$ was an arbitrary equilibrium payoff, 
it follows that $v^* \leq \attain \qconv v_0(\pi_0)$.
\end{proof}

\begin{proof}[Proof of Proposition \ref{prop:nece}]
	If $v_0$ is quasi-concave then $\qconv v_0 = v_0$.
	If $v^* > \qconv v_0$ then $\att\neq \{\pi_0\}$.
	Lemma \ref{lemma:attain} thus implies $\gamma_0<1$.
\end{proof}

\begin{proof}[Proof of Proposition \ref{prop:gap}]
	Since $\att$ is $\subseteq$-decreasing in $\gamma_0$,
	$\attain \qconv v_0$ is weakly decreasing.
	Hence (\emph{i}) follows from Theorem \ref{prop:value}.
	(\emph{iii}) is a corollary of Proposition \ref{prop:first-best},
	and (\emph{ii}) is a corollary of Proposition \ref{prop:IT} and Assumption \ref{ass:sep}.
\end{proof}

\subsection{Applications and extensions}

\begin{proof}[Proof of Proposition \ref{prop:kickstarter}]
We are interested in equilibria in which the event is backed for sure.
If the project gets backed, the expected utility for a receiver with beliefs $\pi$  can be written as
\begin{align}
	\int\limits_\Theta  u_{R}(a_r,\theta) \,d\pi(\theta)
			& = \theta_\pi \log(w + \rho_0 a_r) + (1-\theta_\pi) \log(w-a_r),
\end{align}
where $\theta_\pi := \int_\Theta \theta \,d\pi(\theta)$.
The first order condition for an interior maximum is thus
\begin{align}
	\theta_\pi \dfrac{\rho}{w + \rho a_r} - (1-\theta_\pi) \dfrac{1}{w-a_r} = 0.
\end{align}
This condition yields the best response function
\begin{align}\label{eqn:kick-proof-A}
	\BR(\pi) = \min\left\{ 0,\:  \dfrac{w}{\rho}(\theta_\pi (1 + \rho)  - 1)\right\}.
\end{align}

If $\rho > 1$, then $\int_\Theta \theta \,d\pi_0(\theta) = 1/2 > 1/(1+\rho)$.
This means that receivers are willing to make pledge a positive amount 
$\BR(\pi_0) = w(\rho-1)/2\rho >0$ without any information transmission. 
In this cases, it suffices to have $n  > 2\rho\eta/w(\rho-1)$ 
to have a babbling equilibriun in which the project gets backed. 
The interesting case is when $\rho<1$, so that $\BR(\pi_0) = 0$.
In this case, receivers need to be persuaded to make a positive pledge. 

Let $n' = \min\{k\in\Natural \:\: k \geq n\rho^2/4 \}$
and $\gamma= n'/n \geq \rho^2/4$.
Using $\rho<1$, it can be shown that $n'\leq n-1$ as long as $n\geq 2$.
From the analysis of Example \ref{eg:cont} in Section \ref{sec:posterior},
it follows that there is a $\gamma$-attainable belief $\pi_\gamma$
such that 
\begin{align}\label{eqn:kick-proof-B}
	\int\limits_\Theta \theta \,d\pi_\gamma(\theta)
		= \dfrac{1}{1+\sqrt\gamma}.
\end{align}
Using the strategy that $\gamma$-attains this belief, 
the sender can guarantee a total investment greater than 
\begin{align}
n' \BR(\pi_\gamma) 
	& = n' \dfrac{w}{\rho}\left( \dfrac{1+\rho}{1+\sqrt\gamma}  - 1 \right)
	  \geq \dfrac{n\rho^2}{4} \dfrac{w}{\rho}\left( \dfrac{1+\rho}{1+\rho/2}  - 1 \right)
	  	\nonumber \\
	& = n \cdot \dfrac{w\rho}{4}\cdot  \dfrac{\rho}{2+\rho} 
	  > n \cdot \dfrac{w\rho^2}{12}
	  \geq \dfrac{12 \eta}{w\rho^2} \cdot \dfrac{w\rho^2}{12} = \eta,
\end{align}
where the first equality follows from (\ref{eqn:kick-proof-A}) and (\ref{eqn:kick-proof-B}),
the first inequality from the definition of $n'$ and $\gamma$, 
the second equality from simple algebra,
the second inequality from $\rho < 1$,
and the last inequality from $n \geq 12\eta/w\rho^2$.
Since the project gets backed for sure, the sender has no incentive to deviate 
and the proposed strategy profile constitutes a perfect Bayesian equilibrium. 
\end{proof}


\begin{proof}[Proof of Theorem \ref{prop:commitment}]
($\geq$)
	Let $u_S^* = \attain \conv v_0(\pi_0)$.
	There exists some $\hat{p}\in \att$ such that $u_S^* = \conv v_0(\hat{p})$.
	By Caratheodory's theorem there exist 
	beliefs $p_1,\ldots,p_K$ and weights $\mu\in \Delta^K$
	with $K\leq n_\Theta+1$ such that $\hat{p} = \sum_{k=1}^K \mu_k p_k$,
	and $u_S^* = \sum_{k=1}^K \mu_k v_0(p_k)$.
	Under Assumption \ref{ass:sep}, there exist actions $a_1,\ldots,a_K$
	such that $v_0(p_k) = \tilde{u}_S(a_k)$ and $a_k\in\BR(p_k)$.
	The result then follows from Proposition 1 in \cite{KG}.
	
($\leq$) 
	Let $u_S^*$ be the expected payoff to the sender from a
	commitment protocol $(\bm\alpha,\bm\beta,\mu)$. 
	Note that the arguments from the proof of Lemma \ref{lemma:symmetric}
	can be applied to commitment protocols. 
	Hence, we can assume without loss of generality that $(\bm\alpha,\bm\beta,\mu)$
	is symmetric. 
	Also, since there is no incentive compatibility constraint for the sender, 
	we can assume without loss of generality that 
	$U_S(\alpha_r(m_r,\omega)) = v_0(\beta_r(m_r))$
	for every receiver $r$, $\omega\in[0,1]$ and every message $m_r$ such that $\mu_r(m_r)>0$.

	For every message profile $\mess$ such that $\mu(\mess)>0$,
	let $R(\mess)$ be a set consisting of exactly $n_0$ receivers 
	such that 
	$v_0(\beta_r(m_r))\geq v_0(\beta_{r'}(m_{r'}))$
	for every $r\in R(\mess)$ and $r\not\in R(\mess)$.
	Note that 
	\begin{align}\label{eqn:KG-proof}
		u_S^* 
			& = \int\limits_\Theta \int\limits_{M_0\times M} 
				\dfrac{1}{n_0} \sum_{r\in R(\mess)} 
					v_0(\beta_r(m_r)) \,dx(\mess|\theta)\,d\pi_0(\theta) \nonumber\\
			& \leq \int\limits_\Theta \int\limits_{M_0\times M} 
				\conv v_0(\bar\pi(\mess)) \,dx(\mess|\theta)\,d\pi_0(\theta) 
			\leq \conv v_0(\bar\pi),
	\end{align}
	where $\bar\pi(\mess) := \sum_{r\in R(\mess)} \beta_r(m_r) /n_0$,
	$\bar\pi := \int_\Theta \int_{M_0\times M} 
		\bar\pi(\mess) \,dx(\mess|\theta)\,d\pi_0(\theta)$.
	The first inequality follows from the definition of $\conv$,
	and the second one from Jensen's inequality and the concavity of $\conv v_0$.
	Using a communication strategy analogous to the communication strategy $\mu'$
	from the proof of Theorem \ref{prop:value}, it can be shown
	that $\bar\pi$ is $\gamma_0$-attainable. 
	Hence, (\ref{eqn:KG-proof}) implies that $u_S^* \leq \att\conv v_0(\pi_0)$.
	Since $u_S^*$ was arbitrary, it follows that $v^{**}\leq \att\conv v_0(\pi_0)$.
\end{proof}

\begin{proof}[Proof of Proposition \ref{prop:CS}]
	The proof is constructive. 
	Let $\mathscr{I}_K$ be the partition of $[0,1)$ into $K$ intervals
	of the form $E_k = [(k-1)/K,k/K)$, $k=1,\ldots,K$.
	Let $\bar\theta_k = (2k - 1)/2K$ be the midpoint of the $k$-th interval. 
	Let $k(\theta)$ denote the block of $\mathscr{I}_K$ containing $\theta$.
	And let $k^*(\theta) = \argmin\{ (\theta+b-\bar\theta_k)^2 \:|\: k=1,\ldots,K\}$.

	Consider the tuple $(\bm\alpha,\bm\beta,\mu)$ described as follows.
	The audience is randomly divided into a target set $T\subseteq R$ 
	consisting of exactly $n_0$ receivers, and $R\setminus T$.
	When the state equals $\theta$, 
	receivers $r \not\in T$ are sent the message $m_r = k(\theta)$,
	while, receivers $r\in T$ are always sent the message $m_r=k^*(\theta)$.
	In other words, 
	receivers $r \not\in T$ are truthfully told which block of $\mathscr{I}_K$ contains $\theta_0$, 
	while receivers in $T$ are always told that $\theta_0 \in E_{k^*(\theta_0)}$,
	whether this is true or not. 
	The update rules $\bm\beta$ are derived using Bayes' rule.
	Receivers choose $\alpha_r(m_r) = \Exp{\theta_0 | m_r}$.
	
	Fix $K$.
	Since $n_0$ is also fixed, 
	there exists $\bar{n}$ such that whenever $n\geq \bar{n}$
	$\Exp{\theta_0 | m_r = k} \in E_k$,
	and $k^*(\theta) = \argmin \{ (\theta+b-\Exp{\theta_0 | k})^2 \:|\: k=1,\ldots,K\}$.
	For such values of $n$, the proposed tuple is an equilibrium
	and it yields 
	\begin{align}
		u_S^* = - \left( \theta_0 + b - \Exp{\theta_0 | k^*(\theta_0)} \right)^2.
	\end{align}
	If $\theta_0 + b\leq 1$, then  $\theta_0 + b \in E_{k^*(\theta)}$
	and therefore $u_S^* \geq - 1/K$.
	If not, then $k^*(\theta) = K$ and $u_S^* \geq - (b +  1/K)^2$.
\end{proof}

\end{document}